\documentclass[11pt]{article}

\usepackage{amsthm}
\usepackage{graphicx} 
\usepackage{array} 

\usepackage{nicefrac}
\usepackage{amsmath, amssymb, amsfonts, verbatim}
\usepackage{hyphenat,epsfig,subcaption,multirow}

\usepackage[usenames,dvipsnames]{xcolor}
\usepackage[ruled]{algorithm2e}

\usepackage{tcolorbox}
\tcbuselibrary{skins,breakable}
\tcbset{enhanced jigsaw}

\usepackage[eulergreek]{sansmath}

\usepackage[compact]{titlesec}

\definecolor{DarkRed}{rgb}{0.5,0.1,0.1}
\definecolor{DarkBlue}{rgb}{0.1,0.1,0.5}

\usepackage{nameref}
\definecolor{ForestGreen}{rgb}{0.1333,0.5451,0.1333}
\definecolor{Red}{rgb}{0.9,0,0}
\usepackage[linktocpage=true,
	pagebackref=true,colorlinks,
	linkcolor=DarkRed,citecolor=ForestGreen,
	bookmarks,bookmarksopen,bookmarksnumbered]
	{hyperref}
\usepackage[noabbrev,nameinlink]{cleveref}
\crefname{property}{property}{Property}
\creflabelformat{property}{(#1)#2#3}
\crefname{equation}{eq}{Eq}
\creflabelformat{equation}{(#1)#2#3}

\usepackage{bm}
\usepackage{url}
\usepackage{xspace}
\usepackage[mathscr]{euscript}

\usepackage{tikz}
\usetikzlibrary{arrows}
\usetikzlibrary{arrows.meta}
\usetikzlibrary{shapes}
\usetikzlibrary{backgrounds}
\usetikzlibrary{positioning}
\usetikzlibrary{decorations.markings}
\usetikzlibrary{patterns}
\usetikzlibrary{calc}
\usetikzlibrary{fit}
\usetikzlibrary{snakes}

\usepackage{mdframed}

\usepackage[noend]{algpseudocode}
\makeatletter
\def\BState{\State\hskip-\ALG@thistlm}
\makeatother

\usepackage{cite}
\usepackage{enumitem}

\usepackage[margin=1in]{geometry}

\newtheorem{theorem}{Theorem}
\newtheorem{lemma}{Lemma}[section]
\newtheorem{proposition}[lemma]{Proposition}

\newtheorem{claim}[lemma]{Claim}
\newtheorem{fact}[lemma]{Fact}

\newtheorem*{claim*}{Claim}
\newtheorem*{proposition*}{Proposition}
\newtheorem*{lemma*}{Lemma}
\newtheorem*{problem*}{Problem}
\newtheorem*{theorem*}{Theorem}

\newtheorem{mdresult}{Result}
\newenvironment{Theorem}{\begin{mdframed}[backgroundcolor=lightgray!40,topline=false,rightline=false,leftline=false,bottomline=false,innertopmargin=2pt]\begin{mdresult}}{\end{mdresult}\end{mdframed}}

\newtheorem{definition}[lemma]{Definition}

\newtheorem{problem}{Problem}

\newtheorem{remark}[lemma]{Remark}
\newtheorem*{remark*}{Remark}
\newtheorem{observation}[lemma]{Observation}

\crefname{lemma}{Lemma}{Lemmas}
\crefname{claim}{Claim}{Claims}

\allowdisplaybreaks

\renewcommand{\qed}{\nobreak \ifvmode \relax \else
      \ifdim\lastskip<1.5em \hskip-\lastskip
      \hskip1.5em plus0em minus0.5em \fi \nobreak
      \vrule height0.75em width0.5em depth0.25em\fi}

\newcommand{\rs}{Ruzsa-Szemer\'{e}di\xspace}

\newcommand{\tvd}[2]{\ensuremath{\norm{#1 - #2}_{\textnormal{tvd}}}}
\newcommand{\Ot}{\ensuremath{\widetilde{O}}}
\newcommand{\eps}{\ensuremath{\varepsilon}}
\newcommand{\Paren}[1]{\Big(#1\Big)}
\newcommand{\Bracket}[1]{\Big[#1\Big]}
\newcommand{\bracket}[1]{\left[#1\right]}
\newcommand{\paren}[1]{\ensuremath{\left(#1\right)}\xspace}
\newcommand{\card}[1]{\left\vert{#1}\right\vert}
\newcommand{\Omgt}{\ensuremath{\widetilde{\Omega}}}

\newcommand{\norm}[1]{\ensuremath{\|#1\|}}

\newcommand{\expect}[1]{\Exp\bracket{#1}}

\newcommand{\set}[1]{\ensuremath{\left\{ #1 \right\}}}
\newcommand{\poly}{\mbox{\rm poly}}
\newcommand{\polylog}{\mbox{\rm  polylog}}

\newcommand{\alg}{\ensuremath{\mathcal{A}}\xspace}

\DeclareMathOperator*{\Exp}{\ensuremath{{\mathbb{E}}}}
\DeclareMathOperator*{\Prob}{\ensuremath{\textnormal{Pr}}}
\renewcommand{\Pr}{\Prob}
\newcommand{\EX}{\Exp}
\newcommand{\Ex}{\Exp}

\newcommand{\etal}{{\it et al.\,}}

\newenvironment{tbox}{\begin{tcolorbox}[
		enlarge top by=5pt,
		enlarge bottom by=5pt,
		 boxsep=0pt,
                  left=4pt,
                  right=4pt,
                  top=10pt,
                  arc=0pt,
                  boxrule=1pt,toprule=1pt,
                  colback=white
                  ]
	}
{\end{tcolorbox}}

\newcommand{\rv}[1]{\ensuremath{\textnormal{\textsf{#1}}}\xspace}
\newcommand{\rA}{\rv{A}}
\newcommand{\rB}{\rv{B}}
\newcommand{\rC}{\rv{C}}
\newcommand{\rD}{\rv{D}}
\newcommand{\rE}{\rv{E}}

\newcommand{\supp}[1]{\ensuremath{\textnormal{\text{supp}}(#1)}}
\newcommand{\distribution}[1]{\ensuremath{\textnormal{\text{dist}}(#1)}}

\newcommand{\kl}[2]{\ensuremath{\mathbb{D}(#1~||~#2)}}
\newcommand{\II}{\ensuremath{\mathbb{I}}}
\newcommand{\HH}{\ensuremath{\mathbb{H}}}
\newcommand{\mi}[2]{\ensuremath{\II(#1 \,; #2)}}
\newcommand{\en}[1]{\ensuremath{\HH(#1)}}

\newcommand{\itfacts}[1]{\Cref{fact:it-facts}-(\ref{part:#1})\xspace}

\newcommand{\SI}{\ensuremath{\textnormal{\textsf{SI}}}\xspace}
\newcommand{\estar}{\ensuremath{e^{\star}}}
\newcommand{\istar}{\ensuremath{i^{\star}}}

\newcommand{\prot}{\ensuremath{\pi}}
\newcommand{\Prot}{\ensuremath{\Pi}}

\newcommand{\sstar}{\ensuremath{s^{\star}}}
\newcommand{\tstar}{\ensuremath{t^{\star}}}

\newcommand{\dist}{\ensuremath{\mathcal{D}}}

\newcommand{\unif}{\ensuremath{\mathcal{U}}}

\newcommand{\GRS}{\ensuremath{G^{\mathsf{RS}}}\xspace}
\newcommand{\ERS}{\ensuremath{E^{\mathsf{RS}}}\xspace}
\newcommand{\MRS}{\ensuremath{M^{\mathsf{RS}}}\xspace}

\newcommand{\ureach}{\ensuremath{\textnormal{\textsf{unique-reach}}}\xspace}

\newcommand{\UR}{\ensuremath{\textnormal{\textsf{UR}}}\xspace}
\newcommand{\distUR}{\ensuremath{\dist}_{\UR}}
\newcommand{\protUR}{\ensuremath{\prot}_{\UR}}
\newcommand{\ProtUR}{\ensuremath{\Prot}_{\UR}}

\newcommand{\rProtUR}{\ensuremath{\rProt}_{\UR}}

\newcommand{\IC}[2]{\ensuremath{\textnormal{\textsf{IC}}_{#2}(#1)}\xspace}
\newcommand{\CC}[2]{\ensuremath{\textnormal{\textsf{CC}}_{#2}(#1)}\xspace}

\newcommand{\rX}{\rv{X}}
\newcommand{\rY}{\rv{Y}}
\newcommand{\rR}{\rv{R}}

\newcommand{\rProt}{\ensuremath{\mathsf{\Pi}}}
\newcommand{\rS}{\rv{S}}
\newcommand{\rT}{\rv{T}}

\newcommand{\distSI}{\ensuremath{\dist}_{\ensuremath{\textnormal{\textsf{SI}}}}\xspace}

\newcommand{\protSI}{\ensuremath{\prot}_{\ensuremath{\textnormal{\textsf{\SI}}}}\xspace}
\newcommand{\ProtSI}{\ensuremath{\Prot}_{\ensuremath{\textnormal{\textsf{SI}}}}\xspace}
\newcommand{\rProtSI}{\ensuremath{\rv{\Prot}}_{\ensuremath{\textnormal{\textsf{SI}}}}\xspace}

\renewcommand{\SI}{\ensuremath{\textnormal{\textsf{set-intersection}}}\xspace}

\newcommand{\restar}{\rv{e}^{\star}}

\newcommand{\streach}{\ensuremath{\textnormal{\textsf{st-reachability}}}\xspace}

\newcommand{\ST}{\ensuremath{\textnormal{\textsf{ST}}}}

\newcommand{\distST}{\ensuremath{\dist_{\ST}}}

\newcommand{\distURinv}{\ensuremath{\overleftarrow{\dist}}_{\!\UR}}

\newcommand{\protST}{\prot_{\ST}}

\newcommand{\Hinv}{\ensuremath{\overleftarrow{H}}}

\newcommand{\Einv}{\ensuremath{\overleftarrow{E}}}

\newcommand{\Vstar}{\ensuremath{V^{\star}}}
\newcommand{\Ustar}{\ensuremath{U^{\star}}}

\newcommand{\rsstar}{\ensuremath{\rv{s}^{\star}}}
\newcommand{\rtstar}{\ensuremath{\rv{t}^{\star}}}

\newcommand{\rZ}{\rv{Z}}

\newcommand{\EinvA}{\overleftarrow{{E}}_{\!\!A}}
\newcommand{\EinvB}{\overleftarrow{{E}}_{\!\!B}}

\newcommand{\rEinvB}{\overleftarrow{\rv{E}}_{\!\!B}}

\newcommand{\ristar}{\rv{i}^{\star}}

\newcommand{\ureachinv}{\ensuremath{\overleftarrow{\ureach}}}

\newcommand{\Ginv}{\ensuremath{\overleftarrow{G}}}

\renewcommand{\protSI}{\ensuremath{\prot_{\textnormal{\textsf{SI}}}}}

\renewcommand{\leq}{\leqslant}
\renewcommand{\geq}{\geqslant}

\title{Near-Quadratic Lower Bounds for Two-Pass \\ Graph Streaming Algorithms}
\author{Sepehr Assadi\footnote{(\texttt{sepehr.assadi@rutgers.edu}) Department of Computer Science, Rutgers University. Part of this work was done while the author was a
postdoctoral researcher at Princeton University and was supported in part by the Simons Collaboration on Algorithms
and Geometry. \medskip} \and 
Ran Raz\footnote{(\texttt{ran.raz.mail@gmail.com}) Department of Computer Science, Princeton University. Research supported by the Simons Collaboration on Algorithms and Geometry, by a Simons Investigator Award, and by the National Science Foundation grants No. CCF-171477 and CCF-2007462.}}

\date{}

\begin{document}
\maketitle

\pagenumbering{roman}


\begin{abstract}
	We prove that any \textbf{two-pass} graph streaming algorithm for the $s$-$t$ reachability problem in $n$-vertex directed graphs requires 
	\textbf{near-quadratic} space of $n^{2-o(1)}$ bits. As a corollary, we also obtain near-quadratic space lower bounds for several other fundamental  problems
	including maximum bipartite matching and (approximate) shortest path  in undirected graphs. 
	
	\medskip
	Our results collectively imply that a wide range of graph problems admit essentially no {non-trivial}  streaming algorithm 
	even when two passes over the input is allowed. Prior to our work, such impossibility results were only known for single-pass streaming algorithms, and the best two-pass lower bounds only ruled out $o(n^{7/6})$ space algorithms, leaving open
	a large gap between (trivial) upper bounds and lower bounds. 
	
\end{abstract}

\clearpage

\setcounter{tocdepth}{3}
\tableofcontents

\clearpage

\pagenumbering{arabic}
\setcounter{page}{1}


\section{Introduction}\label{sec:intro}

Graph streaming algorithms process the input graph with $n$ known vertices by making one or a few passes over the sequence
of its unknown edges (given in an arbitrary order) and using a limited memory (much smaller than the input size which is $O(n^2)$ for a graph problem). 
In recent years, graph streaming algorithms and lower bounds for numerous problems have been studied extensively. In particular, we now have a relatively clear picture of the powers and limitations of \emph{single-pass} algorithms. 
With a rather gross oversimplification, this can be stated as follows: 
\begin{itemize}[label=$-$, leftmargin=20pt]
	\item The \emph{exact} variant of most graph problems of interest are intractable: There are $\Omega(n^2)$ space lower bounds for maximum matching and minimum
	 vertex cover~\cite{FeigenbaumKMSZ05,GoelKK12}, (directed) reachability and topological sorting~\cite{HenzingerRR98,FeigenbaumKMSZ05,ChakrabartiG0V20}, shortest path
	 and diameter~\cite{FeigenbaumKMSZ05,FeigenbaumKMSZ08}, minimum or maximum cut~\cite{Zelke11}, maximal independent set~\cite{AssadiCK19,CormodeDK19}, dominating set~\cite{EmekR14,AssadiKL16}, and many others. 
	
	\item On the other hand, \emph{approximate} variants of many graph problems are tractable: There are $\Ot(n) := O(n \cdot \polylog{(n)})$ space algorithms (often referred to as \emph{semi-streaming} algorithms) for 
	approximate (weighted) matching and vertex cover~\cite{FeigenbaumKMSZ05,EpsteinLMS11,CrouchS14,PazS17}, spanner
	 computation and approximation for distance problems~\cite{ElkinZ04,FeigenbaumKMSZ08,Baswana08,Elkin11}, cut or spectral sparsifiers and approximation for cut problems~\cite{AhnG09,AhnGM13,KelnerL11,KapralovLMMS14}, large independents 
	 sets~\cite{HalldorssonHLS10,CormodeDK19}, graph coloring~\cite{AssadiCK19,BeraCG19}, and 
	 approximate dominating set~\cite{EmekR14,AssadiKL16}, among others\footnote{It should be noted that, in contrast, determining the \emph{best} approximation ratio possible for many of these problems have remained elusive and is an active area of research.}. 
\end{itemize}

Recent years have also witnessed a surge of interest in designing \emph{multi-pass} graph
streaming algorithms (see, e.g.~\cite{FeigenbaumKMSZ05,McGregor05,EggertKS09,KonradMM12,Kapralov13,KaleT17, BeckerKKL17,AhnGM12b,SarmaGP11,KapralovW14,McGregorVV16,CormodeJ17,BeraC17,GamlathKMS19,ChakrabartiG0V20,ChangFHT20}); see, e.g.,~\cite{FeigenbaumKMSZ05,McGregor14} for discussions on practical applications of multi-pass streaming
algorithms in particular in obtaining I/O-efficiency. These results suggest that allowing even just one more pass over the input greatly enhances the capability of the algorithms. For instance, while computing the \emph{exact} global or $s$-$t$ minimum cut in undirected graphs requires $\Omega(n^2)$ space
in a single pass~\cite{Zelke11}, perhaps surprisingly, one can solve both problems in only \emph{two} passes with $\Ot(n)$ and $\Ot(n^{5/3})$ space, respectively~\cite{RubinsteinSW18} (see also~\cite{MukhopadhyayN19} for an $O(\log{n})$-pass algorithm for  weighted minimum cut). Qualitatively similar separations are known for numerous other  problems such as triangle counting~\cite{CormodeJ17,BulteauFKP16} (with two passes),
approximate matching~\cite{GoelKK12,Kapralov13,McGregor05,GamlathKMS19}  (with $O(1)$ passes),  maximal
independent set~\cite{AssadiCK19,CormodeDK19,GhaffariGKMR18} (with $O(\log\log{n})$ passes), approximate dominating set~\cite{ChakrabartiW16,Har-PeledIMV16,AssadiKL16} (with $O(\log{n})$ passes), and exact
shortest path~\cite{FeigenbaumKMSZ08,ChangFHT20} (with $O(\sqrt{n})$ passes).  

Despite this tremendous progress, the general picture for the abilities and limitations of {multi-pass} algorithms is not so clear even when we focus on \emph{two-pass} algorithms. What other problems beside minimum cut admit non-trivial two-pass streaming algorithms? For instance, can we obtain similar results for directed versions of these problems? What about closely related problems such as maximum bipartite matching or  not-so-related problems such as shortest path? 
Currently, none of these problems admit any non-trivial 
two-pass streaming algorithm, while known lower bounds only rule out algorithms with $o(n^{7/6})$ space~\cite{FeigenbaumKMSZ08,GuruswamiO13,ChakrabartiG0V20} leaving a considerable gap between upper and lower bounds
 (see~\cite{AssadiCK19b} for a discussion on the current landscape of 
multi-pass graph streaming lower bounds and the challenges in improving them). 

\subsection{Our Contributions}\label{sec:results}

We present near-quadratic space lower bounds for \emph{two-pass} streaming algorithms for several fundamental graph problems including reachability,  bipartite matching, and shortest path. 

\paragraph{Reachability and related problems in \underline{directed} graphs.} We prove the following lower bound for the reachability problem in directed graphs. 

\begin{Theorem}[Formalized in~\Cref{thm:reach}]\label{res:reach}
	Any two-pass streaming algorithm (deterministic or randomized) that given an $n$-vertex directed graph $G=(V,E)$ with two designated vertices $s,t \in V$ can determine whether 
	or not $s$ can reach $t$ in $G$ requires $\Omega(\frac{n^2}{2^{\Theta(\sqrt{\log{n}})}})$ space. 
\end{Theorem}

The reachability problem is one of the earliest problems studied in the graph streaming model~\cite{HenzingerRR98}. 
Previously, Henzinger~\etal~\cite{HenzingerRR98} and Feigenbaum~\etal~\cite{FeigenbaumKMSZ08} proved an $\Omega(n^2)$ space lower bound
for this problem for single-pass algorithms, and Guruswami and Onak~\cite{GuruswamiO13} gave an $\Omgt_p(n^{1+\nicefrac{1}{(2p+2)}})$ lower bound for $p$-pass algorithms which translates to $\Omgt(n^{7/6})$ space for two-pass algorithms; this lower
bound was recently extended to random-order streams by Chakrabarti~\etal~\cite{ChakrabartiG0V20}. Note that the \emph{undirected} version of this problem has a simple $O(n\log{n})$ space algorithm in one pass by maintaining a spanning forest of the input graph (see, e.g.~\cite{FeigenbaumKMSZ05}). 

Using standard reductions, our results in this part can  be extended to several other related problems on directed graphs such as estimating number of vertices reachable from a given source or approximating minimum feedback arc set, studied 
in~\cite{HenzingerRR98} and~\cite{ChakrabartiG0V20}, respectively.

\paragraph{Matching and cut problems.} We have the following lower bound for bipartite matching. 
\begin{Theorem}[Formalized in~\Cref{thm:matching}]\label{res:matching}
	Any two-pass streaming algorithm (deterministic or randomized) that given an $n$-vertex undirected bipartite graph $G=(L\sqcup R,E)$ can determine whether 
	or not $G$ has a perfect matching requires $\Omega(\frac{n^2}{2^{\Theta(\sqrt{\log{n}})}})$ space. 
\end{Theorem}

Maximum matching problem is arguably the most studied problem in the graph streaming model. However, the main focus on this problem so far has been on approximation algorithms and not much is known 
for exact computation of this problem, beside that it can be done in $\Ot(k^2)$ space in a single pass where $k$ is size of the maximum matching~\cite{ChitnisCEHMMV16} (for the perfect matching problem, this gives an $O(n^2)$ space algorithm which is 
the same as storing the entire input). Previously, Feigenbaum~\etal~\cite{FeigenbaumKMSZ05} and Chitnis~\etal~\cite{ChitnisCHM15} proved an $\Omega(n^2)$ space lower bound for single-pass algorithms for this problem and 
Guruswami and Onak~\cite{GuruswamiO13} extended the lower bound to $\Omgt_p(n^{1+\nicefrac{1}{(2p+2)}})$ for $p$-pass algorithms. 

Both the perfect matching problem and the $s$-$t$ reachability problem are simpler versions of the $s$-$t$ minimum cut problem in directed graphs. 
As such, our lower bounds imply that even though the $s$-$t$ minimum cut problem can be solved in undirected 
graphs in $\Ot(n^{5/3})$ space and two passes~\cite{RubinsteinSW18}, its directed version requires $n^{2-o(1)}$ space in two passes (for any multiplicative approximation). 
Previously, Assadi~\etal~\cite{AssadiCK19b} proved a lower bound of $\Omega(n^2/p^5)$ for $p$-pass algorithms for the \emph{weighted} $s$-$t$ minimum cut problem (with exponential-in-$p$ weights); for the unweighted problem, the previous best lower bound was still $\Omgt_p(n^{1+\nicefrac{1}{(2p+2)}})$. 

\paragraph{Shortest path problem.} Finally, we also prove a lower bound for the shortest path problem. 

\begin{Theorem}[Formalized in~\Cref{thm:sp}]\label{res:sp}
	Any two-pass streaming algorithm (deterministic or randomized) that given an undirected graph $G=(V,E)$ and two designated vertices $s,t \in V$, can output the length
	of the shortest $s$-$t$-path  in $G$ requires $\Omega(\frac{n^2}{2^{\Theta(\sqrt{\log{n}})}})$ space. 
	The lower bound continues to hold even for approximation algorithms with approximation ratio  better than $\nicefrac{9}{7}$. 
\end{Theorem}

Shortest path problem have also been extensively studied in graph streaming literature. For single-pass streams, the focus has been on maintaining \emph{spanners} (subgraphs that preserve pairwise distances approximately)
which allow for obtaining algorithms with different space-approximation tradeoffs~\cite{ElkinZ04,FeigenbaumKMSZ08,Baswana08,Elkin11} (starting from $2$-approximation in $O(n^{3/2})$ space
to $O(\log{n})$ approximation in $O(n)$ space), which are known to be almost tight~\cite{FeigenbaumKMSZ08}. For multi-pass algorithms, $\Ot(n)$ space algorithms are known for $(1+\eps)$-approximation with $\poly(\log{n},{1}/{\eps})$ 
passes~\cite{HenzingerKN16,BeckerKKL17}, and exact algorithms with $O(\sqrt{n})$ passes~\cite{ChangFHT20}. On the lower bound front, an $\Omega(n^2)$ space lower bound is known for single-pass 
algorithms~\cite{FeigenbaumKMSZ08} and $\Omgt_{p}(n^{1+\nicefrac{1}{2p+2}})$ for $p$-pass algorithms~\cite{GuruswamiO13}) (for exact answer or even some small approximation $\approx (2p+4)/(2p+2)$); 
a  stronger lower bound of $\Omega(n^{1+\nicefrac{1}{2p}})$ was proven earlier in~\cite{FeigenbaumKMSZ08} for algorithms that need to \emph{find} the shortest path itself.

\medskip

Our results show that a wide range of graph problems including directed reachability, cut and matching, and shortest path problems,  admit essentially no {non-trivial} two-pass streaming algorithms: modulo
the $n^{o(1)}$-term in our bounds, the best one could do to  is to simply store the entire stream in $O(n^2)$ space and solve the problem at the end using any offline algorithm.

\subsection{Our Techniques}\label{sec:techniques}

We prove our main lower bound for the $s$-$t$ reachability problem; the other lower bounds then follow easily from this using standard ideas. 

It helps to start the discussion with the lower bounds in~\cite{FeigenbaumKMSZ08,GuruswamiO13,ChakrabartiG0V20}. These lower bounds work with \emph{random graphs} 
wherein $s$ can reach $\Theta(\sqrt{n})$ {random} vertices $S$ and $t$ is \emph{independently} reachable from $\Theta(\sqrt{n})$ random vertices $T$; thus, by Birthday Paradox,  $s$-$t$ reachability can have either answer 
with constant probability. One then shows that to determine this, the algorithm needs to ``find'' $S$ and $T$ explicitly. The final part is then to 
use ideas from \emph{pointer chasing}  problems~\cite{PapadimitriouS84,NisanW91,PonzioRV99,ChakrabartiCM08,Yehudayoff16} to prove a lower bound for this task. The particular space-pass tradeoff  
is then determined as follows: $(i)$ as a streaming algorithm can find the $p$-hop neighborhood of $s$ and $t$ in $p$ passes (by BFS), $S$ and $T$ need to be $(p+1)$-hop away from $s$ and $t$; $(ii)$ 
as we are working with random graphs, to achieve the bound of $O(\sqrt{n})$ on size of $S$ and $T$, we need the degree of the graph to be $O(n^{{1}/{2(p+1)}})$, leading to an $O(n^{1+1/(2p+2)})$ space lower bound for $p$-pass algorithms.
We note that the \emph{limit} of these approaches based on random graphs seem to be $\Ot(n^{3/2})$; see~\cite[Section 5.2]{ChakrabartiG0V20}.

Our lower bound takes a different route and works with ``more structured'' graphs.
We start with 
proving a \emph{single-pass} streaming lower bound for an ``{algorithmically easier}'' variant of the reachability problem. In this problem, we are promised that $s$ can
reach a \emph{unique} vertex $\sstar$ {chosen uniformly at random} from a set $U$ of $n^{1-o(1)}$ vertices and the goal is to ``find'' this vertex. Previous lower bounds~\cite{FeigenbaumKMSZ08,GuruswamiO13,ChakrabartiG0V20} already imply that if our goal was to determine the identity of $\sstar$ \emph{exactly}, we
need $\Omega(n^2)$ space. In this paper, we prove a stronger lower bound that an $n^{2-o(1)}$-space single-pass algorithm essentially cannot even change
the distribution of $\sstar$ from uniform over $U$. The proof of this part is based on information theoretic arguments that rely on ``embedding'' multiple instances of the
set intersection problem (see~\Cref{sec:si}) inside a \emph{\rs} (RS) graph (see~\Cref{sec:rs}), and proving a new lower bound for the set intersection problem. 

We remark that our new lower bound for set intersection is related to the recent lower bound of~\cite{AssadiCK19b} 
with a subtle technical difference that is explained in~\Cref{sec:si} and~\Cref{rem:internal-external}. We also note that RS graphs have been used extensively  for proving graph streaming
lower bounds~\cite{GoelKK12,Kapralov13,Konrad15,AssadiKLY16,AssadiKL17,CormodeDK19} starting from~\cite{GoelKK12}, but this is their first application to the $s$-$t$ reachability problem. 

In the next part of the argument, we construct a family of graphs in which the $s$-$t$ reachability is determined by existence of a single edge $(\sstar,\tstar)$ in the graph, where 
$\sstar$ is the unique vertex reachable from $s$ in a large set $U$ and $\tstar$ is the unique vertex that can reach $t$ in a large set $W$ (see~\Cref{fig:streach} for an illustration). 
By exploiting our lower bound in the first part, we show that a $n^{2-o(1)}$-space algorithm cannot properly ``find'' 
the pairs $\sstar$ and $\tstar$ in the first pass. 
We then argue that this forces the algorithm to effectively ``store'' all the edges between $U$ and $W$ in the second pass to determine if $(\sstar,\tstar)$ is an edge of the graph, leading to an $n^{2-o(1)}$ space lower bound. 

\begin{remark*}[\textbf{More than two passes?}]
	The intermediate ``simpler'' problem we considered in our proofs (part one above) is only hard in one pass (see~\Cref{sec:ureach}) and thus our lower bound proof does not directly go beyond two passes. 
	However, it \underline{appears} that our techniques can be extended to multi-pass algorithms to prove lower bounds of the type $n^{1+\Omega(\nicefrac{1}{p})}$ space for $p$-pass algorithms which are slightly better in terms 
	of dependence on $p$ in the exponent compared to~\cite{FeigenbaumKMSZ08,GuruswamiO13,ChakrabartiG0V20}. Nevertheless, as unlike the case for two-pass algorithms,  it is no longer clear whether such
	bounds are the ``right'' answer to the problems at hand, we opted to not pursue this direction in this paper. 
\end{remark*}

\subsection{Subsequent Work}\label{sec:subsequent}

Since the conference publication of this work in~\cite{AssadiR20}, our result has inspired several follow-ups. Firstly,~\cite{ChenKPSSY21a} significantly strengthened our approach to prove $n^{2-o(1)}$-space lower bounds 
for the problems considered in this paper in any $o(\sqrt{\log{n}})$ passes. In a nutshell, their techniques can be seen as recursive construction that can ``hide'' a large subset of vertices
from streaming algorithms as opposed to only a pair of vertices $\sstar,\tstar$ discussed in our techniques. Moreover,~\cite{Assadi22} built on the ideas in our work as well as~\cite{ChenKPSSY21a} to prove a two-pass lower bound for approximating matchings up to some (small) constant factor (see also~\cite{KonradN21} that give a two-pass lower bound for a restricted family of algorithms that only compute a greedy matching in their first pass but then can be arbitrary in their second pass). Finally,~\cite{ChenKPSSY21b} also used similar ideas to prove two-pass lower bounds for performing random walks
in directed graphs.

In general, there has  been a rapidly growing body of work on multi-pass graph streaming lower bounds in the last couple of years~\cite{ChakrabartiG0V20,AssadiKSY20,AssadiR20,ChenKPSSY21a,AssadiN21,ChenKPSSY21b,Assadi22,AssadiKZ22,KolPSY23,ChenKPSSY23}, 
and we refer the interested reader to these papers for more details.


\section{Preliminaries}

\paragraph{Notation.} For any integer $t \geq 1$, we use $[t] := \set{1,\ldots,t}$. For any $k$-tuple $X = (X_1,\ldots,X_k)$ and integer $i \in [k]$, we define $X^{<i} := (X_1,\ldots,X_{i-1})$. 

Throughout the paper, we use `sans serif' letters to denote random variables (e.g., $\rA$) , and the corresponding normal letters to denote their values (e.g. $A$). For brevity and to avoid the clutter in notation, 
in conditioning terms which involve assignments to random variables, we directly use the value of the random variable (with the same letter), e.g., write $\rB \mid A$ instead of $\rB \mid \rA = A$. 

For random variables $\rA,\rB$, we use $\en{\rA}$ and $\mi{\rA}{\rB} := \en{\rA} - \en{\rA \mid \rB}$ to denote the Shannon entropy and mutual information, respectively. Moreover, for two distributions 
$\mu,\nu$ on the same support, $\tvd{\mu}{\nu}$ denotes the total variation distance, and $\kl{\mu}{\nu}$ is the KL-divergence. A summary of basic information theory facts
that we use in our proofs appear in~\Cref{app:info}.

\paragraph{Concentration bounds:} We will use the following variant of Chernoff bound (see, e.g.~\cite{DubhashiP09}). 

\begin{proposition}[Chernoff bound]\label{prop:chernoff}
	Let $\rX_1,\ldots,\rX_n$ be $n$ independent random variables with values in $[0,1]$ and $\rX := \sum_{i=1}^{n} \rX_i$. Then, for any $b > 0$, 
	\begin{align*}
		\Pr\paren{\card{\rX - \expect{\rX}} \geq b} \leq 2 \cdot \exp\paren{-\frac{b^2}{2n}}. 
	\end{align*}
\end{proposition}

\subsection{Communication Complexity and Information Complexity}\label{sec:cc-ic}

We work with the two-party communication model of Yao~\cite{Yao79} (with some slightly non-standard aspects mentioned in~\Cref{sec:streach}). See the excellent textbook by Kushilevitz and Nisan~\cite{KushilevitzN97} for 
an overview of communication complexity. 

Let $P: \mathcal{X} \times \mathcal{Y} \rightarrow \mathcal{Z}$ be a relation.  Alice receives an input $X\in \mathcal{X}$ and Bob receives $Y \in \mathcal{Y}$, where $(X,Y)$ are chosen from a
distribution $\dist$ over $\mathcal{X} \times \mathcal{Y}$. We allow players to have access to both public and private randomness. 
They communicate with each other by exchanging messages according to some \emph{protocol} $\prot$. Each message in $\prot$
depends only on the private input and random bits of the player sending the message, the already communicated messages, and the public randomness. 
At the end, one of the players outputs an answer  $Z$ such that $Z \in P(X,Y)$. For any protocol $\prot$, we use $\Prot := \Prot(X,Y)$ to denote the messages \emph{and} the public randomness used 
by $\prot$ on the input $(X,Y)$. 

We now define two measures of ``cost'' of a protocol. 

\begin{definition}[Communication cost]
  The \emph{communication cost} of a protocol $\prot$, denoted by $\CC{\prot}{}$, is the \emph{worst-case} length of the messages
  communicated between Alice and Bob in the protocol. 
\end{definition}

\begin{definition}[Information cost]
	The \emph{(internal) information cost} of a protocol $\prot$, when the inputs $(X,Y)$ are drawn from a distribution $\dist$, is $\IC{\prot}{\dist}:=\mi{\rProt}{\rX \mid \rY}{} + \mi{\rProt}{\rY \mid \rX}{}$. 
\end{definition}

The internal information cost (introduced by~\cite{BarakBCR10}; see also~\cite{ChakrabartiSWY01,Bar-YossefJKS02,BarakBCR10,BravermanR11,BravermanEOPV13}) measures the average amount of information each player learns about the input of the other player
by observing the transcript of the protocol. As each bit of communication cannot reveal  more than one bit of information, the internal information
cost of a protocol lower bounds its communication cost~\cite{BravermanR11}. 

\paragraph{Communication complexity and streaming.} There is a standard connection between the communication cost of any protocol $\prot$ for a communication problem $P(X,Y)$
and the space of any streaming algorithm that can determine the value of $P(X,Y)$ on a stream $X \circ Y$ (see~\Cref{prop:stream}); we use this connection to establish 
our streaming lower bounds. 
\subsection{\rs Graphs}\label{sec:rs}

A graph $\GRS=(V,E)$ is a called an $(r,t)$-\rs graph (RS graph for short) iff its edge-set $E$ can be partitioned into $t$ \emph{induced matchings} $\MRS_1,\ldots, \MRS_t$, each of size $r$. We further define 
an $(r,t)$-RS \emph{digraph} as a directed {bipartite} graph $\GRS=(L,R,E)$ obtained by directing every edge of a bipartite $(r,t)$-RS graph from $L$ to $R$.

We use the original 
construction of RS graphs due to Ruzsa and Szemer\'{e}di~\cite{RuszaS78} based on the existence of large sets of integers with no $3$-term arithmetic progression, proven by Behrend~\cite{Behrend46}. We note that
there are multiple other constructions with different parameters (see, e.g.~\cite{FischerLNRRS02,AlonMS12,GoelKK12,FoxHS15} and references therein) but the following construction works best for our purpose. 

\begin{proposition}[\!\!\cite{RuszaS78}]\label{prop:rs}
	For infinitely many integers $N$, there are $(r,t)$-RS digraphs with $N$ vertices on each side of the bipartition and parameters $r = \frac{N}{e^{\Theta(\sqrt{\log{N}})}}$ and $t = {N}/{3}$. 
\end{proposition}



\newcommand{\rsigma}{\sansmath{\sigma}}
\newcommand{\rSigma}{\mathsf{\Sigma}}

\section{A New Lower Bound for the Set Intersection Problem}\label{sec:si}

One key ingredient of our paper is a new lower bound for the set intersection problem, defined formally as follows. 

\begin{problem}[$\SI$]\label{def:si}
	The $\SI$ problem is a  two-player communication problem in which Alice and Bob are given sets $A$ and $B$ from $[m]$, respectively, with the promise that there exists a \textbf{unique}
element $\estar$ such that $\set{\estar} = A \cap B$. The goal is  to find the \textbf{target element} $\estar$ using back and forth communication (i.e., in the two-way communication model). 
\end{problem}

The \SI problem is closely related to the well-known \emph{set disjointness} problem. It is in fact straightforward to prove an $\Omega(m)$ lower bound on the communication complexity of $\SI$ using a simple reduction from the 
set disjointness problem. However, in this paper, we are interested in an algorithmically simpler variant of this problem which we define below. 

\begin{definition}\label{def:eps-si}
	Let $\dist$ be a distribution of  inputs $(A,B)$ for $\SI$ (known to both players). A protocol $\prot$ \textbf{internal $\eps$-solves} $\SI$ over  $\dist$ iff 
	\underline{at least} one of the following holds:  
	\begin{align}
		&\EX_{\rProt,\rA}{\tvd{\distribution{\restar \mid \Prot,A}}{\distribution{\restar \mid A}}} \geq \eps \quad \textnormal{or} \quad 
		\EX_{\rProt,\rB}{\tvd{\distribution{\restar \mid \Prot,B}}{\distribution{\restar \mid B}}} \geq \eps, 
	\end{align}
	where all variables are defined with respect to the distribution $\dist$ and the internal randomness of $\prot$ (recall that $\Prot$ includes the transcript and the public randomness). 
\end{definition}

\Cref{def:eps-si} basically states that a protocol can internal $\eps$-solve the $\SI$ problem iff 
the transcript of the protocol can change the distribution of the target element $\estar$ \emph{from the perspective of Alice or Bob} by at least $\eps$ in the total variation distance on average.

Our definition is inspired but different from $\eps$-solving in~\cite{AssadiCK19b} (which we call \emph{external} $\eps$-solving to avoid ambiguity) which required the transcript to change the distribution of the target 
element by $\eps$ \emph{from the perspective of an external observer} (who only sees the transcript but not the inputs of players). 
More formally, external $\eps$-solving of $\SI$ over a distribution $\mu$, as defined in~\cite{AssadiCK19b}, requires the protocol $\prot$ to have the following property (compare this with~\Cref{def:eps-si}), 
\[
	\EX_{\rProt}{\tvd{\distribution{\restar \mid \Prot}}{\distribution{\restar}}} \geq \eps.
\]

The previous work in~\cite{AssadiCK19b} has shown that there is a distribution $\mu$ such that any protocol that external $\eps$-solves $\SI$ over $\mu$ has information cost $\Omega(\eps^2 \cdot m)$.  
This however does \emph{not} imply a lower bound  for the internal $\eps$-solving problem. This is because, \emph{in principle}, these two tasks can be different. For instance, $(i)$ a protocol 
that reveals the entire set of Alice, changes the distribution of target for Bob dramatically but not so much for an external observer; or $(ii)$ a protocol that reveals all the elements that are 
neither in $A$ nor in $B$ changes the distribution of the target for an external observer by a lot but does not change the distribution for either of the players at all. 

We prove the following lower bound on the information cost of internal $\eps$-solving of $\SI$. 

\begin{theorem}\label{thm:eps-si}
	There is a distribution $\distSI$ for $\SI$ over the universe $[m]$ such that: 
	\begin{enumerate}[label=(\roman*)]
	\item For any $A$ or $B$ sampled from $\distSI$, both $\distribution{\restar \mid A}$ and $\distribution{\restar \mid B}$ are uniform distributions on $A$ and $B$, each of size $m/4$, respectively. 
	\item For any $\eps \in (0,1)$, any protocol $\prot$ that internal $\eps$-solves the $\SI$ problem over the distribution $\distSI$ (\Cref{def:eps-si}) has  internal information cost $\IC{\prot}{\distSI} = \Omega(\eps^2 \cdot m)$.  
	\end{enumerate}
\end{theorem}

In the rest of this section, we give a simple proof of~\Cref{thm:eps-si} by relating it to known results for exact solving of the $\SI$ problem (our proof idea can also be used to give a much simpler proof of the
lower bound of~\cite{AssadiCK19b} for the external $\eps$-solving of $\SI$). However, the reader may skip this section entirely and jump to the next one for the proof of our graph streaming lower bounds directly (the only
result used from this section that is used outside is~\Cref{thm:eps-si}).

\subsection{The Lower Bound for $\eps$-Internal Solving of Set Intersection}\label{sec:proof-eps-si}
\subsection*{Background on Set Intersection} Consider the following distribution: 
\begin{tbox}
	\textbf{Distribution $\distSI$.} An input distribution $(A,B)$ to $\SI$ over the universe $[m]$.
\begin{enumerate}[label=$(\roman*)$]
	\item Sample two \emph{disjoint} sets $\tilde{A}$ and $\tilde{B}$ of size $m/4-1$ each uniformly at random from $[m]$. 
	\item Sample $\estar \in [m] \setminus \tilde{A} \cup \tilde{B}$ uniformly at random and let $A := \tilde{A} \cup \set{\estar}$ and $B := \tilde{B} \cup \set{\estar}$. 
\end{enumerate}
\end{tbox}
\noindent
It is well-known that $\distSI$ is a ``hard'' distribution for $\SI$~\cite{Razborov92}. In particular, the following result follows from the lower bound of~\cite{JayramKS03}. 
\begin{proposition}[cf. \cite{JayramKS03}]\label{prop:si}
	Any protocol $\prot$ that \underline{finds} the target element in $\SI$ over  $\distSI$ with probability of success at least $2/3$ has  information cost $\IC{\prot}{\distSI} = \Omega(m)$. 
\end{proposition}
A remark on the origin of~\Cref{prop:si} is in order. Firstly, if one only cares about proving a communication cost lower bound for $\SI$, then the proposition follows from the classical results
on the communication complexity of the set disjointness problem~\cite{KalyanasundaramS92,Razborov92,Bar-YossefJKS02}. At the same time, the lower bound on the information cost 
in~\Cref{prop:si} does \emph{not} follow from the standard information cost lower bounds for set disjointness (see, e.g.~\cite{Bar-YossefJKS02,BravermanM13,BravermanGPW13,WeinsteinW15}). 
This is because all aforementioned works bound the information cost of protocols for set disjointness on \emph{non-intersecting} distributions, while $\distSI$ is crucially intersecting (one can fix this however by applying the recent
 result of~\cite{GoosJPW18}, itself based on~\cite{BravermanW15}). As such, one needs to apply the result
of~\cite{JayramKS03} on the information cost of set disjointness protocols on intersecting distributions (see~\cite{AssadiCK19b} for  details). 

\subsection*{Proof of~\Cref{thm:eps-si}}

Let $\protSI$ be any protocol that internal $\eps$-solves $\SI$ over $\distSI$ and without loss of generality, let us assume it does so by changing the distribution of $\estar$ from the perspective of Alice, i.e., 
\begin{align}
	\EX_{\rProt,\rA}{\tvd{\distribution{\restar \mid \Prot,A}}{\distribution{\restar \mid A}}} \geq \eps. \label{eq:protSI}
\end{align}
Define the following two constants: 
\begin{equation}\label{eq:constants}
\begin{aligned}
	\gamma_1 &:= \text{minimum of $1/100$ and the constant of $\Omega$-notation of~\Cref{prop:si}}, \\ 
	\gamma_2 &:= \text{$100 \cdot \log{(1/\gamma_1)}$ for the constant $\gamma_1$ above}.
\end{aligned}
\end{equation}
We design a protocol $\prot$ based on $\protSI$ that finds the target element $\estar$ of a given instance $(A,B)$ of $\SI$ sampled from $\distSI$ with probability at least $2/3$. The general idea of the protocol 
is to use a ``self-reducibility'' property by turning each fixed instance of $\distSI$ to many independent instances sampled uniformly from $\distSI$ and run the protocol $\protSI$ on each one. Alice then collects the elements 
that she deems to be ``more likely'' to be the target element in each of these instances into a set $T$ much smaller than $A$, and communicates $T$ directly to Bob. Our analysis  uses the fact that 
the distribution of $\estar$ from the perspective of Alice in each of these instances was sufficiently far from uniform and prove that this smaller set $T$ still should contain the target element $\estar$ with a constant probability. 

\begin{tbox}
	\textbf{Protocol $\prot$.} A protocol for (exact) solving $\SI$ on $\distSI$ using $\protSI$ as a black-box. 
	\begin{enumerate}[label=$(\roman*)$]
	\item Define the parameters: 
	\[
	k := \frac{32}{\eps^2} \cdot \ln{(\frac{100\,\gamma_2}{\gamma_1})}, \quad t := \frac{\gamma_1}{\gamma_2} \cdot \frac{m}{2}, \quad  \tau := \paren{\frac{1}{2}+\frac{\eps}{4}} \cdot k.
	\]
	\item For $i = 1$ to $k$ times: 
	\begin{enumerate}
		\item Sample a permutation $\sigma_i$ of $[m]$ uniformly at random using \underline{public randomness}.
		\item Run $\protSI$ on the inputs $(\sigma_i(A),\sigma_i(B))$ (where for $S \in \set{A,B}$, $\sigma(S) := \set{\sigma(i) \mid i \in S}$). 
		\item Let $\Prot_i$ be the transcript of $\protSI$. Alice computes $\distribution{\estar \mid \Prot_i,\sigma_i(A)}$ and let $S_i$ be the 
		top half of elements in $A$ with largest probability in this distribution.  
	\end{enumerate}
	\item For any element $e \in A$, Alice computes $c_e$ as the number of $i \in [k]$ such that $\sigma_i(e)$ appears in $S_i$. Let $T$ be the set of elements $e \in A$ where $c_e > \tau $. 
	
	\item If $\card{T} > t $, Alice reports `fail'. Otherwise, Alice communicates the set $T$ directly to Bob (appended with arbitrary elements from $A$ until its size reaches exactly $t$ if it is smaller). 
	Bob returns the element in $B \cap T$ or reports `fail' if no element was found. 
\end{enumerate}
\end{tbox}
\begin{observation}\label{obs:si-key}
	The distribution of each $(\sigma_i(A),\sigma_i(B))$  is $\distSI$, and $\set{(\sigma_i(A),\sigma_i(B))}_{i \in [k]}$ are mutually independent. 
\end{observation}
To see the proof of~\Cref{obs:si-key} consider the following alternative way of sampling from the distribution $\distSI$: Let $A=\set{1,\ldots,m/4}$ and $B= \set{m/4,\ldots,(m/2)-1}$, and then apply a random permutation $\sigma$ and 
return $\sigma(A),\sigma(B)$. Under this view,~\Cref{obs:si-key} is immediate.

\Cref{obs:si-key} identifies the role of $\sigma_i$'s in the protocol -- they basically ``re-randomize'' the input instance so that we can try an \emph{independent} version of $\protSI$ on the original input multiple times. 

The following two claims are the key parts of the proof. 

\begin{claim}\label{clm:prot-estar}
For any $i \in [k]$, $\Pr\paren{\sigma_i(\estar) \in S_i} \geq \dfrac{1+\eps}{2}.$
\end{claim}
\begin{proof}
	Considering~\Cref{obs:si-key}, any choice of $\Prot_i,\sigma_i(A)$, fixes $\distribution{\sigma_i(\estar) \mid \Prot_i,\sigma_i(A)}$ and the distribution of $\sigma_i(\estar)$ over $\sigma_i(A)$ now is exactly this. 
	At the same time, $\distribution{\sigma_i(\estar) \mid \sigma_i(A)}$ is uniform over $\sigma_i(A)$ by construction of $\distSI$. 
	We can now use~\Cref{prop:uniform-tvd} (on probability mass of larger half of distribution based on the total variation distance from uniform) to obtain,  
	\begin{align*}
		\Pr\paren{\sigma_i(\estar) \in S_i} &= \Ex_{\Prot_i,\sigma_i(A)}\bracket{\Pr\paren{\sigma_i(\estar) \in S_i \mid \Prot_i,\sigma_i(A)}} \\
		&\geq \Ex_{\Prot_i,\sigma_i(A)}\bracket{\frac{1}{2}+\frac{1}{2} \cdot \tvd{\distribution{\sigma_i(\estar) \mid \Prot_i,\sigma_i(A)}}{\distribution{\sigma_i(\estar) \mid \sigma_i(A)}}} 
		\tag{by~\Cref{prop:uniform-tvd} and the choice of $S_i$ in the protocol} \\
		&\geq \frac{1}{2}+\frac{\eps}{2},
	\end{align*}
	where the final equation is by the guarantee of $\protSI$ for internal $\eps$-solving in~\Cref{eq:protSI}. 
\end{proof}

\begin{claim}\label{clm:prot-e}
	For any $i \in [k]$ and $e \neq \estar \in A$, $\Pr\paren{\sigma_i(e) \in S_i} \leq \dfrac{1}{2} + \dfrac{1}{\card{A}}$. 
\end{claim}
\begin{proof}
	Consider any fixing of $(\sigma_i(A),\sigma_i(B))$. Over the randomness of $\sigma_i$, any element $e \neq \estar \in A$ is mapped to $\sigma_i(A) \setminus \sigma_i(\estar)$ \emph{uniformly at random} even 
	conditioned on $\sigma_i(A),\sigma_i(B)$ (because $\sigma_i(e) \notin B$). At the same time, fixing of $(\sigma_i(A),\sigma_i(B))$  also fixes $\Prot_i$ and subsequently the set $S_i$. As such, 
	\begin{align*}
		\Pr\paren{\sigma_i(e) \in S_i} \leq \frac{\card{S_i}}{\card{A}-1} = \frac{1}{2} \cdot \frac{1}{1-\frac{1}{\card{A}}} \leq \frac{1}{2} \cdot \paren{1+\frac{2}{\card{A}}} = \frac{1}{2} + \frac{1}{\card{A}},
	\end{align*}
	concluding the proof. 
\end{proof}
\Cref{clm:prot-estar,clm:prot-e} imply that $\estar$ tends to appear more frequently in the sets $S_i$ for $i \in [k]$ compared to other elements $e \in A$. 
Hence, by picking the threshold $\tau$ for adding the elements into $T$ carefully, we are able to ensure that $\estar \in T$ with large probability, while not too many 
other elements belong to $T$ as well. We formalize this in the following claim whose proof is a basic application of Chernoff bound. 

\begin{claim}\label{clm:prot-chernoff}
	With probability at least $9/10$, $\estar \in T$ and $\card{T} \leq t$. 
\end{claim}
\begin{proof}
	For any element $e \in A$, let $X(e)$ denote the number indices $i \in [k]$ where $\sigma_i(e)$ appears in $S_i$, i.e., $X(e) = c_e$. By~\Cref{clm:prot-estar}, we have 
	\[
	\expect{X(\estar)} \geq \paren{\frac{1}{2} + \frac{\eps}{2}} \cdot k = \tau + \frac{\eps}{4} \cdot k \tag{by the choice of $\tau = \paren{\frac{1}{2}+\frac{\eps}{4}} \cdot k$}.
	\]
	Moreover, by~\Cref{obs:si-key}, the choices across the $k$ instances $(\sigma_i(A),\sigma_i(B))$ are independent. As such, by Chernoff bound (\Cref{prop:chernoff}), 
	\begin{align*}
		\Pr\paren{X(\estar) \leq \tau} \leq \Pr\Paren{\card{X(\estar) - \expect{X(\estar)}} \geq \frac{\eps}{4} \cdot k} \leq 2 \cdot \exp\paren{-\frac{\eps^2}{32} \cdot k } \leq 
		\frac{\gamma_1}{100 \cdot \gamma_2} \leq \frac{1}{100}. \tag{by the choice of $k = \dfrac{32}{\eps^2} \cdot \ln{(\dfrac{100 \cdot \gamma_2}{\gamma_1}})$ in $\prot$, and since $\gamma_1 \leq 1$ and $\gamma_2 \geq 1$ in~\Cref{eq:constants}}
	\end{align*}
	This means that $\estar \in T$ with probability at least $99/100$.
	
	At the same time, for any $e \neq \estar$, by~\Cref{clm:prot-e}, 
	\[
	\expect{X(e)} \leq \paren{\frac{1}{2}+\frac{4}{m}} \cdot k \leq \tau - \frac{\eps}{4} \cdot k+1, \tag{by the choice of $\tau = \paren{\dfrac{1}{2}+\dfrac{\eps}{4}} \cdot k$}
	\]
	where we assumed that $m > 4k$ as otherwise $\eps^2 \cdot m = \Theta(1)$ and thus~\Cref{thm:eps-si} hold vacuously. 
	
	As before, by Chernoff bound, 
	\begin{align*}
		\Pr\paren{X(e) > \tau} \leq \Pr\Paren{\card{X(e) - \expect{X(e)}} \geq \frac{\eps}{4} \cdot k} \leq  \frac{\gamma_1}{100\cdot \gamma_2}. \tag{as calculated above}
	\end{align*}
	As such, 
	\begin{align*}
		\Ex\card{T} \leq 1+\sum_{e\neq \estar \in A} \expect{X(e)} \leq 1+\frac{m}{4} \cdot \frac{\gamma_1}{100 \cdot \gamma_2}, 
	\end{align*}
		and thus, by Markov bound, with probability less than $1/100$, $\card{T} \leq \frac{\gamma_1}{\gamma_2} \cdot \frac{m}{2} = t$, as desired.  
	The claim now follows from a union bound. 
\end{proof}

We need yet another claim that bounds the information cost of the protocol $\prot$. The proof  uses the fact that~\Cref{obs:si-key} ensures each $(\sigma_i(A),\sigma_i(B),\Prot_i)$ are distributed
as in $\distSI$ (and hence we can bound its information cost) and the final set $T$ communicated by Alice is sufficiently small. 

\begin{claim}\label{clm:prot-info}
	$\IC{\prot}{\distSI} \leq k \cdot \IC{\protSI}{\distSI} + \dfrac{\gamma_1}{2} \cdot m$. 
\end{claim}
\begin{proof}
	Let $\rProt_1,\ldots,\rProt_k$ and $\sigma_1,\ldots,\sigma_k$ denote the random variables for, respectively, the transcript of the protocol $\protSI$ concatenated with its internal public randomness, and the random permutations 
	in iterations $1$ to $k$ of the for-loop in $\prot$. Additionally, let $\rT$ be the random variable for the set $T$ communicated by Alice to Bob. By the definition of internal information cost and the chain rule of mutual information (\itfacts{chain-rule}):
	$\IC{\prot}{\distSI}$ is sum of the following three terms: 
	\begin{align}
		&\mi{\rsigma_1,\ldots,\rsigma_k}{\rA \mid \rB} + \mi{\rsigma_1,\ldots,\rsigma_k}{\rB \mid \rA}; \label{eq:bounding1} \\ 
		&\mi{\rProt_1,\ldots,\rProt_k}{\rA \mid \rB, \rsigma_1,\ldots,\rsigma_k} + \mi{\rProt_1,\ldots,\rProt_k}{\rB \mid \rA, \rsigma_1,\ldots,\rsigma_k}; \label{eq:bounding2} \\
		&\mi{\rT}{\rA \mid \rB,\rProt_1,\ldots,\rProt_k,\rsigma_1,\ldots,\rsigma_k} + \mi{\rT}{\rB \mid \rA,\rProt_1,\ldots,\rProt_k,\rsigma_1,\ldots,\rsigma_k}. \label{eq:bounding3}
	\end{align}
	We bound each term separately below. 
	
	\paragraph{Bounding~\Cref{eq:bounding1}.} This term is zero because the permutations are chosen independent of $(\rA,\rB)$ and thus the given mutual information term is zero by~\itfacts{info-zero}. 
	
	\paragraph{Bounding~\Cref{eq:bounding2}.} Define $\rSigma:=(\rsigma_1,\ldots,\rsigma_k)$. We have, 
	\begin{align*}
	\mi{\rProt_1,\ldots,\rProt_k}{\rA \mid \rB,\rSigma} + \mi{\rProt_1,\ldots,\rProt_k}{\rB \mid \rA,\rSigma} &= \sum_{i=1}^{k} \mi{\rProt_{i}}{\rA \mid \rB,\rSigma,\rProt_{<i}} + \mi{\rProt_{i}}{\rB \mid \rA,\rSigma,\rProt_{<i}} \tag{by chain rule of mutual information (\itfacts{chain-rule})} \\
	&\leq \sum_{i=1}^{k} \mi{\rProt_{i}}{\rA \mid \rB,\rsigma_{i}} + \mi{\rProt_{i}}{\rB \mid \rA,\rsigma_i} \tag{as ${\rProt}_i$ $\perp \rSigma_{<i},\rProt_{<i} \mid \rA,\rB,\sigma_i$ since the conditioning fixes $\rProt_i$  and we can apply~\Cref{prop:info-decrease}} \\
	&= \sum_{i=1}^{k} \mi{\rProt_i}{\rsigma_i(\rA) \mid \rsigma_i(\rB)} + \mi{\rProt_{i}}{\rsigma_i(\rB) \mid \rsigma_i(\rA)} \tag{exactly as above as $\rProt_i$ only depends on the ``inner'' instance defined by $\sigma_i(A)$ and $\sigma_i(B)$} \\
	&= k \cdot \IC{\protSI}{\distSI} \tag{by~\Cref{obs:si-key}}. 
\end{align*}
\paragraph{Bounding~\Cref{eq:bounding3}.} We use $\rSigma$ defined in the previous part here as well:
\begin{align*}
\mi{\rT}{\rA \mid \rB,\rProt_1,\ldots,\rProt_k,\rSigma} + \mi{\rT}{\rB \mid \rA,\rProt_1,\ldots,\rProt_k,\rSigma} &= \mi{\rT}{\rA \mid \rB,\rProt,\rSigma} \tag{as $\rA,\rProt,\rSigma$ fixes $\rT$ and hence second term is zero (\itfacts{info-zero})} \\
&\leq \en{\rT \mid \rB,\rProt,\rSigma} \tag{by definition of mutual information and non-negativity of entropy)} \\
	&\leq \en{\rT} \tag{as conditioning cannot increase entropy (\itfacts{cond-reduce} )} \\
	&\leq \log\!{{m}\choose{t}} \tag{as $\rT$ is supported on $t$-subsets of $[m]$ and by~\itfacts{uniform}} \\
	&\leq t \cdot \log{\paren{\frac{e \cdot m}{t}}} \tag{by standard upper bounds on binomial coefficients} \\
	&\leq \frac{\gamma_1}{\gamma_2} \cdot \frac{m}{2} \cdot \log{\paren{\frac{2e \cdot \gamma_2}{\gamma_1}}} \tag{by the choice of $t$ in the protocol $\prot$} \\
	&\leq \gamma_1 \cdot m \cdot \frac{1}{{{100 \cdot \log{(1/\gamma_1)}}}} \cdot \log{\paren{\frac{1}{{\gamma_1}^3}}} \tag{by the choice of $\gamma_2 = 100 \cdot \log{(1/\gamma_1)} \leq 100 \cdot (1/\gamma_1)$ and since $\gamma_1 \leq 1/100$ in~\Cref{eq:constants}} \\
	&< \frac{\gamma_1}{2} \cdot m,
\end{align*}
by the definition of $\gamma_1$ and~\Cref{prop:si}. Putting the bounds in~\Cref{eq:bounding1},~\Cref{eq:bounding2}, and~\Cref{eq:bounding3} for $\IC{\prot}{\distSI}$ proves the claim. 
\end{proof}

We are now ready to conclude the proof of~\Cref{thm:eps-si}. Firstly, by~\Cref{clm:prot-chernoff}, the probability that the protocol $\prot$ finds the correct index $\estar$ is at least $9/10$, because conditioned 
on the events in~\Cref{clm:prot-chernoff}, Alice simply sends all of $T$ to Bob which includes $\estar$ and Bob will be able to output the answer. As such, by~\Cref{prop:si}, we know that the information cost of $\prot$
has to be $\gamma_1 \cdot m$ at least. Combining this with~\Cref{clm:prot-info} implies that
\[
	\gamma_1 \cdot m \leq \IC{\prot}{\distSI} \leq k \cdot \IC{\protSI}{\distSI} + \dfrac{\gamma_1}{2} \cdot m,
\]
and in turn results in 
\[
	\IC{\protSI}{\distSI} \geq \frac{\gamma_1}{2k} \cdot m. 
\]
which is $\Omega(\eps^2 \cdot m)$ by the choice of $k$ and since $\gamma_1 = \Theta(1)$ in~\Cref{eq:constants} (by~\Cref{prop:si}). \Cref{thm:eps-si} now follows from this and the definition of the distribution $\distSI$ for part $(i)$. \qed


\newcommand{\rM}{\rv{M}}

\section{The Unique-Reach Communication Problem}\label{sec:ureach}

We now start with our main lower bounds. Define the following two-player communication problem. 

\begin{problem}[\ureach]\label{def:ureach}
	The $\ureach$ problem is defined as follows. Consider a digraph $G=(V,E)$ on $n$ vertices where $V := \set{s} \sqcup V_1 \sqcup V_2 \sqcup V_3$ and any edge $(u,v) \in E$ is directed from $s$ to $V_1$ or some $V_i$ to $V_{i+1}$ for $i \in [2]$ (we refer to each 
	$V_i$ as a \emph{layer}). We are {promised} that there is a \textbf{unique} vertex $\sstar$ in the layer $V_3$ reachable from $s$. 
	
	\medskip
	 Alice is given edges in $E$ from $V_1$ to $V_2$, denoted by $E_A$, and Bob is given the remainder of the edges in $E$, denoted by $E_B$ (the partitioning of vertices of $V$ is known to both players). The
	 goal for the players is to \textbf{find} $\sstar$ by Alice sending 
	 a single message to Bob (i.e., in the one-way communication model). 
\end{problem}

It is easy to prove a lower bound of $\Omega(n^2)$ on the one-way communication complexity of $\ureach$ using a reduction from the Index problem. It is also easy to see that this problem can be solved with $O(n\log{n})$ bits of communication, if 
we allow Bob to send a single message to Alice: By the uniqueness promise on $\sstar$, no vertex with out-degree more than one in $V_2$ should be reachable from $s$ and thus Bob can communicate all the remaining edges in $E_B$ 
to Alice. 

Nevertheless, in this paper, we are interested in an algorithmically simpler variant of this problem similar-in-spirit to $\eps$-solving for $\SI$ (\Cref{def:eps-si}).

\begin{definition}\label{def:eps-ureach}
	Let $\dist$ be any distribution  of valid inputs $G=(V,E_A \sqcup E_B)$ for $\ureach$ (known to both players). We say that a protocol $\prot$ \textbf{internal $\eps$-solves} $\ureach$ over $\dist$ iff: 
	\begin{align}
		\EX_{\rProt,\rE_B}{\tvd{\distribution{\rsstar \mid \Prot,E_B}}{\distribution{\rsstar \mid E_B}}} \geq \eps, 
	\end{align}
	where all variables are defined with respect to the distribution $\dist$ and the internal randomness of $\prot$ (recall that $\Prot$ includes the transcript and the public randomness). 
\end{definition}

\Cref{def:eps-ureach} basically states that a protocol can internal $\eps$-solve the problem iff 
the message sent from Alice can change the distribution of the unique vertex $\sstar$ {from the perspective of Bob} by at least $\eps$ in the total variation distance (in expectation over Alice's message and Bob's input).

Our main theorem in this section is the following. 

\begin{theorem}\label{thm:ureach}
	There is a distribution $\distUR$ for $\ureach$ and an integer $b := \frac{n}{2^{\Theta(\sqrt{\log{n}})}}$ with the following properties: 
	\begin{enumerate}[label=(\roman*)]
	\item For any $E_B$ sampled from $\distUR$, $\distribution{\rsstar \mid E_B}$ is a uniform distribution over a subset $\Vstar_3$ of $b$ vertices in the layer $V_3$ of the input graph; 
	\item for any $\eps \in (0,1)$, any \underline{one-way} protocol $\prot$ that internal $\eps$-solves $\ureach$ over the distribution $\distUR$ (\Cref{def:eps-ureach}) has communication cost $\CC{\prot}{} = \Omega(\eps^2 \cdot n \cdot b)$. 
	\end{enumerate}
\end{theorem}

Proof of~\Cref{thm:ureach} is by a reduction from~\Cref{thm:eps-si} using RS graphs (see \Cref{sec:rs}). 

\subsection{Distribution $\distUR$ in~\Cref{thm:ureach}}

To continue, we need to set up some notation. Let $\GRS=(L,R,E)$ be an $(r,t)$-RS digraph with induced matchings $\MRS_1,\ldots,\MRS_t$ as defined in~\Cref{sec:rs}. For each induced matching $\MRS_i$, we assume an arbitrary ordering of edges $e_{i,1},\ldots,e_{i,r}$ in
the matching and for each $j \in [r]$ denote $e_{ij} := (u_{ij},v_{ij})$ for $u_{ij} \in L$ and $v_{ij} \in R$; moreover, we let $L(\MRS_i) := \set{u_{i1},\ldots,u_{ir}}$ and $R(\MRS_i) := \set{v_{i1},\ldots,v_{ir}}$. 
Based on these, we have the following definition: 
\begin{itemize}
	\item For any matching $\MRS_i$ and any set $S \subseteq [r]$, we define $\MRS_i | S$ as the matching in $\GRS$ consisting of the edges $e_{ij} \in \MRS_i$ for all $j \in S$. 
\end{itemize}
\noindent
We are now ready to define our distribution. See \Cref{fig:ureach} for an illustration.

\begin{figure*}[h!]
    \centering
    \begin{subfigure}[t]{0.5\textwidth}
        \centering

\begin{tikzpicture}
	\node[ellipse, black, fill=black!10, draw, line width=1pt, minimum width=40pt, minimum height=100pt] (V1) {};
		\node[ellipse, black, fill=black!0, draw, line width=1pt, minimum width=30pt, minimum height=40pt] (A1) [above=-50pt of V1]{};
			\node [circle, black, fill=Blue!50, line width=1pt, draw, minimum width=1pt, minimum height=1pt] (A11) [above=-15pt of A1]{};
			\node [circle, black, fill=Blue!50, line width=1pt, draw, minimum width=1pt, minimum height=1pt] (A12) [above=-27pt of A1]{};
			\node [circle, black, fill=Blue!50, line width=1pt, draw, minimum width=1pt, minimum height=1pt] (A13) [above=-39pt of A1]{};
		\node[ellipse, black, fill=black!0, draw, line width=1pt, minimum width=30pt, minimum height=40pt] (B1) [below=-50pt of V1]{};
			\node [circle, black, fill=Blue!50, line width=1pt, draw, minimum width=1pt, minimum height=1pt] (B11) [above=-15pt of B1]{};
			\node [circle, black, fill=Blue!50, line width=1pt, draw, minimum width=1pt, minimum height=1pt] (B12) [above=-27pt of B1]{};
			\node [circle, black, fill=Blue!50, line width=1pt, draw, minimum width=1pt, minimum height=1pt] (B13) [above=-39pt of B1]{};
	
	\node[ellipse, black, fill=black!10, draw, line width=1pt, minimum width=40pt, minimum height=100pt] (V2) [right=1cm of V1]{};
		\node[ellipse, black, fill=black!0, draw, line width=1pt, minimum width=30pt, minimum height=40pt] (A2) [above=-50pt of V2]{};
			\node [circle, black, fill=Blue!50, line width=1pt, draw, minimum width=1pt, minimum height=1pt] (A21) [above=-15pt of A2]{};
			\node [circle, black, fill=Blue!50, line width=1pt, draw, minimum width=1pt, minimum height=1pt] (A22) [above=-27pt of A2]{};
			\node [circle, black, fill=Blue!50, line width=1pt, draw, minimum width=1pt, minimum height=1pt] (A23) [above=-39pt of A2]{};
	\node[ellipse, black, fill=black!0, draw, line width=1pt, minimum width=30pt, minimum height=40pt] (B2) [below=-50pt of V2]{};
			\node [circle, black, fill=Blue!50, line width=1pt, draw, minimum width=1pt, minimum height=1pt] (B21) [above=-15pt of B2]{};
			\node [circle, black, fill=Blue!50, line width=1pt, draw, minimum width=1pt, minimum height=1pt] (B22) [above=-27pt of B2]{};
			\node [circle, black, fill=Blue!50, line width=1pt, draw, minimum width=1pt, minimum height=1pt] (B23) [above=-39pt of B2]{};

	\draw[line width=1pt, black]
		(A11.center) -- (A21.center)
		(A12.center) -- (A22.center)
		(A13.center) -- (A23.center)
		
		(B11.center) -- (A21.center)
		(B12.center) -- (A22.center)
		(B13.center) -- (A23.center)
		
		(A11.center) -- (B21.center)
		(A12.center) -- (B22.center)
		(A13.center) -- (B23.center)
		(B11.center) -- (B21.center)
		(B12.center) -- (B22.center)
		(B13.center) -- (B23.center);

	\node (n1) [below=0.25cm of V1]{$L$};
	\node (n2) [below=0.25cm of V2]{$R$};

\end{tikzpicture}
        \caption{A fixed $(3,4)$-RS digraph in the distribution.}
    \end{subfigure}%
    ~ 
    \begin{subfigure}[t]{0.5\textwidth}
        \centering

\begin{tikzpicture}
	\node[ellipse, black, fill=black!10, draw, line width=1pt, minimum width=40pt, minimum height=100pt] (V1) {};
		\node[ellipse, black, fill=black!0, draw, line width=1pt, minimum width=30pt, minimum height=40pt] (A1) [above=-50pt of V1]{};
			\node [circle, black, fill=Yellow, line width=1pt, draw, minimum width=1pt, minimum height=1pt] (A11) [above=-15pt of A1]{};
			\node [circle, black, fill=Yellow, line width=1pt, draw, minimum width=1pt, minimum height=1pt] (A12) [above=-27pt of A1]{};
			\node [circle, black, fill=Yellow, line width=1pt, draw, minimum width=1pt, minimum height=1pt] (A13) [above=-39pt of A1]{};
		\node[ellipse, black, fill=black!0, draw, line width=1pt, minimum width=30pt, minimum height=40pt] (B1) [below=-50pt of V1]{};
			\node [circle, black, fill=ForestGreen!75, line width=1pt, draw, minimum width=1pt, minimum height=1pt] (B11) [above=-15pt of B1]{};
			\node [circle, black, fill=ForestGreen!75, line width=1pt, draw, minimum width=1pt, minimum height=1pt] (B12) [above=-27pt of B1]{};
			\node [circle, black, fill=ForestGreen!75, line width=1pt, draw, minimum width=1pt, minimum height=1pt] (B13) [above=-39pt of B1]{};
	
	\node[ellipse, black, fill=black!10, draw, line width=1pt, minimum width=40pt, minimum height=100pt] (V2) [right=1cm of V1]{};
		\node[ellipse, black, fill=black!0, draw, line width=1pt, minimum width=30pt, minimum height=40pt] (A2) [above=-50pt of V2]{};
			\node [circle, black, fill=Yellow, line width=1pt, draw, minimum width=1pt, minimum height=1pt] (A21) [above=-15pt of A2]{};
			\node [circle, black, fill=Yellow, line width=1pt, draw, minimum width=1pt, minimum height=1pt] (A22) [above=-27pt of A2]{};
			\node [circle, black, fill=Yellow, line width=1pt, draw, minimum width=1pt, minimum height=1pt] (A23) [above=-39pt of A2]{};
			
	\node[ellipse, black, fill=black!0, draw, line width=1pt, minimum width=30pt, minimum height=40pt] (B2) [below=-50pt of V2]{};
			\node [circle, black, fill=ForestGreen!75, line width=1pt, draw, minimum width=1pt, minimum height=1pt] (B21) [above=-15pt of B2]{};
			\node [circle, black, fill=ForestGreen!75, line width=1pt, draw, minimum width=1pt, minimum height=1pt] (B22) [above=-27pt of B2]{};
			\node [circle, black, fill=ForestGreen!75, line width=1pt, draw, minimum width=1pt, minimum height=1pt] (B23) [above=-39pt of B2]{};

	\draw[line width=1pt, black]
		(A12.center) -- (A22.center)
		(A13.center) -- (A23.center)
		
		(B11.center) -- (A21.center)
		(B13.center) -- (A23.center)
		
		(A11.center) -- (B21.center)
		(A13.center) -- (B23.center)

		(B11.center) -- (B21.center)
		(B12.center) -- (B22.center);

	\draw[line width=0.5pt, gray, dashed]
		(A11.center) -- (A21.center)
		(B12.center) -- (A22.center)
		(A12.center) -- (B22.center)
		(B13.center) -- (B23.center);
	
		\node[ellipse, black, fill=black!0, draw, line width=1pt, minimum width=30pt, minimum height=40pt] (C1) [right=1cm of V2]{};
			\node [circle, black, fill=Blue!50, line width=1pt, draw, minimum width=1pt, minimum height=1pt] (C11) [above=-15pt of C1]{};
			\node [circle, black, fill=blue!50, line width=1pt, draw, minimum width=1pt, minimum height=1pt] (C12) [above=-27pt of C1]{};
			\node [circle, black, fill=Blue!50, line width=1pt, draw, minimum width=1pt, minimum height=1pt] (C13) [above=-39pt of C1]{};

	\node [circle, black, fill=blue!50, line width=1pt, draw, minimum width=1pt, minimum height=1pt] (s) [left=1cm of V1]{};
				
	\draw[line width=1pt, black]
		(C11.center) -- (A21.center)
		(C12.center) -- (A22.center);

	\draw[line width=1pt, black]
		(s) -- (A11.center)
		(s) -- (A12.center);

	\draw[line width=1.5pt, blue]
		(s) -- (A12.center)
		(A12.center) -- (A22.center)
		(A22.center) -- (C12.center);

	\node (n1) [below=0.25cm of V1]{$V_1$};
	\node (n2) [below=0.25cm of V2]{$V_2$};
	\node (n3) [below=0.25cm of C1]{$V_3$};
	\node (ns) [below=0.25cm of s]{$s$};

\end{tikzpicture}
        \caption{The  graph $G$ of $\distUR$; dashed edges no longer belong to the graph, and yellow vertices are incident on $\MRS_{\istar}$.}
    \end{subfigure}
   \caption{An illustration of the  input distribution $\distUR$. Here, directions of all edges are from left to right and hence omitted. The marked vertex (blue) in $V_3$ denotes the unique vertex $\sstar$ in this example along with the path connecting $s$
   to $\sstar$.}
\label{fig:ureach}
\end{figure*}
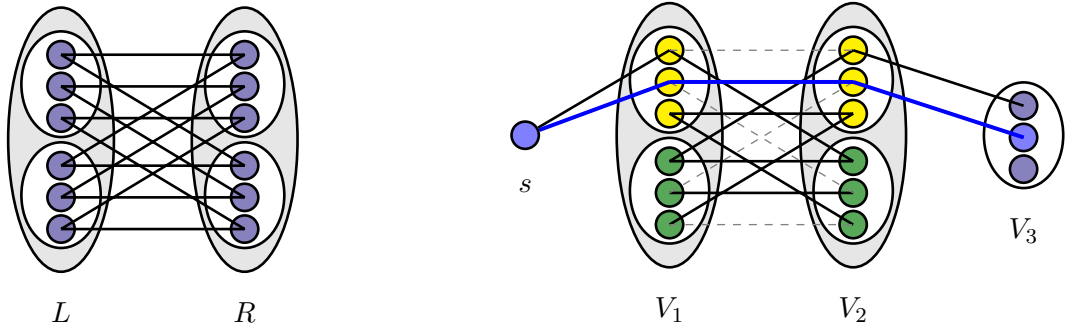

\begin{tbox}
	\textbf{Distribution $\distUR$.} An input distribution on graphs $G=(\set{s} \sqcup V_1 \sqcup V_2 \sqcup V_3, E_A \sqcup E_B)$. 
	
	\begin{enumerate}[label=$(\arabic*)$]
		\item Let $\GRS = (L,R,\ERS)$ be a \emph{fixed} $(r,t)$-RS digraph on $2N$ vertices from~\Cref{prop:rs} with parameters $r = \frac{N}{2^{\Theta(\sqrt{\log{N}})}}$, and $t = \frac{N}{3}$. We note that this graph is known to both players. 
		\item Let $V_1 = L = \set{u_1,\ldots,u_N}$, $V_2 = R = \set{v_1,\ldots,v_N}$, and $V_3$ be $r$ new vertices $\set{w_1,\ldots,w_r}$. 
		\item Sample $t$ \emph{independent} instances $(S_1,T_1),\ldots,(S_t,T_t)$ of $\SI$ on the universe $[r]$ from the distribution $\distSI$ in~\Cref{thm:eps-si}. 
		\item The input $E_A$ to Alice is $E_A:= (\MRS_1|S_1) \cup \ldots \cup (\MRS_t|S_t)$. 
		\item Sample $\istar \in [t]$ \emph{uniformly at random}. 
		\item The input $E_B$ to Bob is the set of edges $(s,u_{\istar j})$ for $j \in T_{\istar}$ and  $(v_{\istar j},w_j)$ for $j \in T_{\istar}$. 
	\end{enumerate}
\end{tbox}

\begin{observation}\label{obs:distUR}
	Several observations are in order:
	\begin{enumerate}[label=$(\roman*)$]
		\item For any $G \sim \distUR$, there is a unique vertex $\sstar$ reachable from $s$ in $V_3$. Moreover, $\sstar = w_{\estar}$ where $\estar \in [r]$ is the unique element in the 
		intersection of $S_{\istar}$ and $T_{\istar}$. \\
		(proof: $\sstar$ is reachable from $s$ through the path $s \rightarrow u_{\istar,\estar} \rightarrow v_{\istar,\estar} \rightarrow w_{\estar} (= \sstar)$ and is the only such vertex 
		by the uniqueness of $\estar$ in $S_{\istar} \cap T_{\istar}$.)
		\item For any $E_B \sim \distUR$, $\distribution{\rsstar \mid E_B}$ is uniform over vertices $w_j \in V_3$ for $j \in T_{\istar}$. \\
		(proof: considering the one-to-one mapping between $\sstar$ and $w_{\estar}$, and since $\distribution{\restar \mid T_{\istar}}$ is uniform over $T_{\istar}$ by~\Cref{thm:eps-si}, $\distribution{\rsstar \mid E_B}$ is also
		 uniform over $w_j$ for $j \in T_{\istar}$.)
		\item In $\distUR$, the index $\istar \in [t]$ is \emph{independent} of the sets $(S_1,T_1),\ldots,(S_t,T_t)$. Moreover, the pairs $(S_1,T_1),\ldots,(S_t,T_t)$ are \emph{mutually independent}. 
		\item The input $E_A$ to Alice in $\distUR$ is uniquely determined by $S_1,\ldots,S_t$, and the input $E_B$ to Bob is determined by $\istar$ and $T_{\istar}$. 
	\end{enumerate}
\end{observation}

\subsection{Proof of \Cref{thm:ureach}}

Let $\protUR$ be any one-way protocol that internal $\eps$-solves $\ureach$ on the distribution $\distUR$. We will prove that $\CC{\protUR}{} = \Omega(\eps^2 \cdot r \cdot t)$ which proves~\Cref{thm:ureach}. 
The argument relies on the following two claims: $(i)$ internal $\eps$-solving of $\ureach$ on $\distUR$ is equivalent to internal $\eps$-solving of $\SI$ on $\distSI$ for the pair $(S_{\istar},T_{\istar})$; and
$(ii)$ the information revealed by $\protUR$ about the instance $(S_{\istar},T_{\istar})$ is \emph{at least $t$ times smaller} than $\CC{\protUR}{}$. Having both these steps, we can then invoke~\Cref{thm:eps-si} to conclude the proof. 

We shall emphasize that this is \emph{not} an immediate reduction from~\Cref{thm:eps-si} as we are aiming to gain an \emph{additional} factor of $t$ in the information cost lower bound for $\protUR$ compared to the lower bound for $\SI$. This  part crucially relies on the fact that $\protUR$ is a one-way protocol and that index $\istar \in [t]$ in the distribution is independent of Alice's input (and thus her message). 

We now present the formal proof. Consider the following protocol $\protSI$ for $\SI$ on the distribution $\distSI$ using $\protUR$ as a subroutine. 

\paragraph{Protocol $\protSI$.} Given an instance $(A,B) \sim \distSI$ on universe $[r]$, Alice and Bob do as follows:  
\begin{enumerate}[label=$(\roman*)$]
	\item Alice and Bob sample $\istar \in [t]$ using \underline{public randomness}. 
	\item Alice sets $S_{\istar} = A$ and samples the remaining sets $S_i$ for $i \neq \istar \in [t]$ independently from $\distSI$ using \underline{private randomness} (this is doable by part $(iii)$ of~\Cref{obs:distUR}). This allows Alice to generate the set $E_A$ of edges for $\protUR$ as in  $\distUR$ (by part $(iv)$ of~\Cref{obs:distUR}). 
	\item Bob sets $T_{\istar} = B$ and creates the set of edges $E_B$ for $\protUR$ as in  $\distUR$ (again doable by part $(iv)$ of~\Cref{obs:distUR} as Bob also knows $\istar$). 
	\item The players then run the protocol $\protUR$ on the input $(E_A,E_B)$ with Alice sending the message in $\protUR$ to Bob. 
\end{enumerate}

We first prove the following  claim. 
\begin{claim}\label{clm:protsi-eps-solves}
	$\protSI$ internal $\eps$-solves $\SI$ on $\distSI$. 
\end{claim}
\begin{proof}
	The distribution of instances $E_A,E_B$ created in the protocol $\protSI$ (using the randomness of the input) is the same as $\distUR$. We can now examine whether $\protSI$ $\eps$-solves 
	$\SI$ as follows: 
	\begin{align*}
		\Ex_{\rProtSI,\rB}{\tvd{\distribution{\restar \mid \ProtSI,B}}{\distribution{\restar \mid B}}} &= \Ex_{\rProtUR,\ristar,\rB}{\tvd{\distribution{\restar \mid \ProtUR,\istar,B}}{\distribution{\restar \mid B}}} \tag{since $\rProtSI = (\rProtUR,\ristar)$ as the only \emph{extra} public randomness of $\protSI$ over $\protUR$ is the choice of $\istar$} \\ 
		&= \Ex_{\rProtUR,\rE_B}{\tvd{\distribution{\rsstar \mid \ProtUR,E_B}}{\distribution{\rsstar \mid E_B}}}, 
	\end{align*}
	since $(\istar,B)$ uniquely identifies $E_B$ and vice versa, and since by~\Cref{obs:distUR}, modulo a renaming, we have $\distribution{\restar \mid B} = \distribution{\sstar \mid E_B}$. The last term above is now
	at least $\eps$ since $\protUR$ internal $\eps$-solves $\SI$, which concludes the proof. 
\end{proof}

We now bound the internal information cost of $\protSI$ which allows us to apply~\Cref{thm:eps-si} and conclude the proof. The proof of this lemma is by a direct-sum style argument. 
We note that these arguments (based on information theory tools) are by now mostly standard in the literature. 
\begin{lemma}\label{lem:ds-ureach}
	$\IC{\protSI}{\distSI} \leq \frac{1}{t} \cdot \CC{\protUR}{}$. 
\end{lemma}

\begin{proof}
	By definition of the internal information cost, 
	\begin{align*}
		\IC{\protSI}{\distSI} = \mi{\rProtSI}{\rA \mid \rB} + \mi{\rProtSI}{\rB \mid \rA} = \mi{\rProtSI}{\rA \mid \rB} \leq \mi{\rProtSI}{\rA};
	\end{align*}
	this is because $\protSI$ is a one-way protocol and thus  $\rProtSI \perp \rB \mid \rA$ and so we can apply~\itfacts{info-zero} for the
	 last equality and~\Cref{prop:info-decrease} 
	for the final inequality. We  now bound the RHS above:
	\begin{align*}
		\mi{\rProtSI}{\rA} &= \mi{\rProtUR,\ristar}{\rA} \tag{as the only additional public randomness of $\protSI$ \emph{over} $\protUR$ is the choice of $\istar$} \\
		&= \mi{\ristar}{\rA} + \mi{\rProtUR}{\rA \mid \ristar} \tag{by the chain rule of mutual information (\itfacts{chain-rule})} \\
		&= \mi{\rProtUR}{\rA \mid \ristar} \tag{as $\ristar \perp \rA$ and so $\mi{\ristar}{\rA} = 0$ by~\itfacts{info-zero}} \\
		&=  \sum_{\istar=1}^{t} \frac{1}{t} \cdot \mi{\rProtUR}{\rA \mid \istar} \tag{as $\istar \in [t]$ is chosen uniformly at random} \\
		&= \frac{1}{t} \cdot \sum_{\istar=1}^{t} \mi{\rProtUR}{\rS_{\istar} \mid \istar} \tag{as $\rA = \rS_{\istar}$ conditioned on $\ristar=\istar$ by definition of the protocol} \\
		&= \frac{1}{t} \cdot \sum_{\istar=1}^{t} \mi{\rProtUR}{\rS_{\istar}} \tag{as the input of Alice and  her one-way message are independent of $\ristar=\istar$ by~\Cref{obs:distUR}} \\
		&\leq \frac{1}{t} \cdot \sum_{\istar=1}^{t}  \mi{\rProtUR}{\rS_{\istar} \mid \rS^{<\istar}} \tag{by~\Cref{prop:info-increase} since $\rS_{\istar} \perp \rS^{<\istar}$ by~\Cref{obs:distUR}} \\
		&= \frac{1}{t} \cdot \mi{\rProtUR}{\rS_1,\ldots,\rS_t}. \tag{by the chain rule of mutual information in~\itfacts{chain-rule}}
	\end{align*}
	We can now bound the $\mi{\rProtUR}{\rS_1,\ldots,\rS_t}$ term above by the communication cost of the protocol $\protUR$ as follows. Define $\rR$ as the public randomness
	of the protocol $\protUR$ and $\rM$ as its message sent from Alice to Bob. We have, 
	\begin{align*}
		\mi{\rProtUR}{\rS_1,\ldots,\rS_t} &= \mi{\rM,\rR}{\rS_1,\ldots,\rS_t} = \mi{\rR}{\rS_1,\ldots,\rS_t} + \mi{\rM}{\rS_1,\ldots,\rS_t \mid \rR} \tag{by the chain rule of mutual information (\itfacts{chain-rule})} \\
		&= \mi{\rM}{\rS_1,\ldots,\rS_t \mid \rR} \tag{as $\rR$ is independent of the input and thus the first term is zero by~\itfacts{info-zero}} \\
		&\leq \log{\card{\supp{\rM}}} = \CC{\protUR}{} \tag{by~\Cref{fact:info-cc}}.  
	\end{align*}
	Plugging the equations above together, we obtain that, 
	\begin{align*}
		\IC{\protSI}{\distSI} \leq \mi{\rProtSI}{\rA} \leq \frac{1}{t} \cdot \mi{\rProtUR}{\rS_1,\ldots,\rS_t} \leq \frac{1}{t} \cdot \CC{\protUR}{},
	\end{align*}
	proving the lemma.  
\end{proof}

We now conclude the proof of~\Cref{thm:ureach}. By~\Cref{clm:protsi-eps-solves}, $\protSI$ internal $\eps$-solves $\SI$ and thus by~\Cref{thm:eps-si}, we  have $\IC{\protSI}{\distSI} = \Omega(\eps^2 \cdot r)$. 
Plugging in this bound in~\Cref{lem:ds-ureach}, we obtain that
\[
	\CC{\protUR}{} = \Omega(\eps^2 \cdot r \cdot t) = \Omega(\eps^2 \cdot \frac{N^2}{2^{\Theta(\sqrt{\log{N}})}}) = \Omega(\eps^2 \cdot \frac{n^2}{2^{\Theta(\sqrt{\log{n}})}}), 
\]
as the number of vertices $n$ in the graph is $O(N)$. Setting $b=r/4 = \frac{n}{2^{\Theta(\sqrt{\log{n}})}}$ now concludes the proof of~\Cref{thm:ureach} (part~$(i)$ of the theorem already follows from~\Cref{obs:distUR}).

\subsection{The Inverse Unique-Reach Problem}\label{sec:ureach-inv} 

In addition to the $\ureach$ problem, we also need another (almost identical) variant of this problem which we call the \emph{inverse} of the $\ureach$ problem, denoted by $\ureachinv$. This problem is basically what one would naturally expect
if we \emph{reverse} the direction of all edges in an instance of $\ureach$ and ask for finding the unique vertex that can now reach the end-vertex $t$ (corresponding to $s$). Formally, we define this problem as follows. 

In $\ureachinv$, we have a digraph $\Ginv=(U,\Einv)$ on $n$ vertices where $U := U_3 \sqcup U_2 \sqcup U_1 \sqcup \set{t}$, all edges of the graph are directed from $U_1$ to $t$ or some $U_{i+1}$ to $U_{i}$ for $i \in [2]$,  
and we are promised that there is a \emph{unique} vertex $\tstar$ in $U_3$ that can reach $t$. The goal is to find this vertex $\tstar$, or rather, internal $\eps$-solve it exactly as in~\Cref{def:eps-ureach}. As before, 
the edges between $U_2$ and $U_1$, denoted by $\EinvA$, are given to Alice, and the remaining edges, denoted by $\EinvB$, are given to Bob. The communication is also one-way from Alice to Bob.  

We also define a hard input distribution for $\ureachinv$, named $\distURinv$, in exact analogy with $\distUR$ for $\ureach$: $\distURinv$ is a distribution over graphs $\Ginv=(U_3 \sqcup U_2 \sqcup U_1 \sqcup \set{t}, \EinvA \sqcup \EinvB)$, 
obtained by sampling a graph $G = (\set{s} \sqcup V_1 \sqcup V_2 \sqcup V_3, E_{A} \sqcup E_{B})$ from $\distUR$, 
setting $U_3 = V_3$, $U_2 = V_2$, $U_1 = V_1$, and $t=s$, and \emph{reversing} the direction of all edges in $E_A$ and $E_B$ to obtain $\EinvA$ and $\EinvB$.


\newcommand{\tdist}{\widetilde{\dist}}
\newcommand{\rO}{\rv{O}}

\newcommand{\rx}{\rv{x}}
\newcommand{\ry}{\rv{y}}

\newcommand{\BB}{\mathcal{B}}

\section{The $\mathbf{st}$-Reachability Communication Problem}\label{sec:streach}

We now define the main two-player communication problem (the setting of this problem is rather non-standard in terms of the communication model). 

\begin{problem}[\streach]
	Consider a digraph $G=(V,E)$ with two designated vertices $s,t$ and $E := E_1 \sqcup E_2 \sqcup E_3$. The goal is to determine whether or not $s$ can reach $t$ in $G$. 
	
	Initially, Alice receives $E_1$ and Bob receives $E_2$ (the vertices $s,t$ are known to both players). 
	Next, Alice and Bob will have one round of communication by Alice sending a message $\Prot_{A1}$ to Bob and Bob responding back with a message $\Prot_{B1}$. 
	At this point, the edges $E_3$ are revealed to \textbf{both} players. Finally, Alice is allowed to send yet another message $\Prot_{A2}$ to Bob (which this time depends on $E_3$ as well)
	and Bob outputs the answer (again also a function of $E_3$). 
\end{problem}

The following theorem is the main result of our paper. 

\begin{theorem}\label{thm:streach}
	For any $\eps \in (n^{-1/2},1/2)$, any communication protocol  for $\streach$ that succeeds with probability  at least $\frac{1}{2}+\eps$ requires $\Omega(\eps^2 \cdot \frac{n^2}{2^{\Theta(\sqrt{\log{n}})}})$ bits of communication. 
\end{theorem}

We note that the $n^{-1/2}$ lower bound on $\eps$ in~\Cref{thm:streach} is not sacrosanct and any term which is $\omega\paren{\frac{\log{n}}{b}}$ still works where $b = \frac{n}{2^{\Theta(\sqrt{\log{n}})}}$ is the parameter in~\Cref{thm:ureach}. 


\subsection{A Hard Distribution for $\streach$}\label{sec:dist-streach}

Recall the distributions $\distUR$ and $\distURinv$ from~\Cref{sec:ureach}. We will use them to define our distribution for $\streach$. See \Cref{fig:streach} for an illustration. 
\begin{tbox}
	\textbf{Distribution $\distST$.} A hard input distribution for the $\streach$ problem. 
	\begin{enumerate}[label=$(\arabic*)$]
		\item Let $V := \set{s} \sqcup V_1 \sqcup V_2 \sqcup V_3 \sqcup U_3 \sqcup U_2 \sqcup U_1 \sqcup \set{t}$ -- each $V_i$ or $U_i$ is called a \emph{layer} of $G$ (this partitioning is known to both players).
		\item Sample the graph $G_1 := (V_3 \sqcup U_3,E_1)$ by picking each edge $(v,u) \in V_3 \times U_3$ independently and with probability half. 
		\item Sample the following {two} graphs \emph{independently}: 
		\begin{enumerate}[label=$(\roman*)$]
			\item $H := (\set{s} \sqcup V_1 \sqcup V_2 \sqcup V_3, E_{A} \sqcup E_{B})$ sampled from the distribution $\distUR$;
			\item $\Hinv := (U_3 \sqcup U_2 \sqcup U_1 \sqcup \set{t}, \EinvA \sqcup \EinvB)$ sampled from the  distribution $\distURinv$. 
		\end{enumerate}
		\item The initial input to Alice and Bob are, respectively, $E_1$ and $E_2 := E_A \cup \Einv_A$, and the input revealed to both players in the
		 second round is $E_3 := E_B \cup \Einv_B$. 
	\end{enumerate}
\end{tbox}

To avoid potential confusion, we should note right away that Bob in the distribution $\distST$ is receiving the input of Alice in $\distUR$ and $\distURinv$.

\begin{figure*}[t!]
    \centering


\begin{tikzpicture}
	\node[ellipse, black, fill=black!10, draw, line width=1pt, minimum width=40pt, minimum height=100pt] (V1) {};
		\node[ellipse, black, fill=black!0, draw, line width=1pt, minimum width=30pt, minimum height=40pt] (A1) [above=-50pt of V1]{};
			\node [circle, black, fill=Yellow, line width=1pt, draw, minimum width=1pt, minimum height=1pt] (A11) [above=-15pt of A1]{};
			\node [circle, black, fill=Yellow, line width=1pt, draw, minimum width=1pt, minimum height=1pt] (A12) [above=-27pt of A1]{};
			\node [circle, black, fill=Yellow, line width=1pt, draw, minimum width=1pt, minimum height=1pt] (A13) [above=-39pt of A1]{};
		\node[ellipse, black, fill=black!0, draw, line width=1pt, minimum width=30pt, minimum height=40pt] (B1) [below=-50pt of V1]{};
			\node [circle, black, fill=ForestGreen!75, line width=1pt, draw, minimum width=1pt, minimum height=1pt] (B11) [above=-15pt of B1]{};
			\node [circle, black, fill=ForestGreen!75, line width=1pt, draw, minimum width=1pt, minimum height=1pt] (B12) [above=-27pt of B1]{};
			\node [circle, black, fill=ForestGreen!75, line width=1pt, draw, minimum width=1pt, minimum height=1pt] (B13) [above=-39pt of B1]{};
	
	\node[ellipse, black, fill=black!10, draw, line width=1pt, minimum width=40pt, minimum height=100pt] (V2) [right=1cm of V1]{};
		\node[ellipse, black, fill=black!0, draw, line width=1pt, minimum width=30pt, minimum height=40pt] (A2) [above=-50pt of V2]{};
			\node [circle, black, fill=Yellow, line width=1pt, draw, minimum width=1pt, minimum height=1pt] (A21) [above=-15pt of A2]{};
			\node [circle, black, fill=Yellow, line width=1pt, draw, minimum width=1pt, minimum height=1pt] (A22) [above=-27pt of A2]{};
			\node [circle, black, fill=Yellow, line width=1pt, draw, minimum width=1pt, minimum height=1pt] (A23) [above=-39pt of A2]{};
	\node[ellipse, black, fill=black!0, draw, line width=1pt, minimum width=30pt, minimum height=40pt] (B2) [below=-50pt of V2]{};
			\node [circle, black, fill=ForestGreen!75, line width=1pt, draw, minimum width=1pt, minimum height=1pt] (B21) [above=-15pt of B2]{};
			\node [circle, black, fill=ForestGreen!75, line width=1pt, draw, minimum width=1pt, minimum height=1pt] (B22) [above=-27pt of B2]{};
			\node [circle, black, fill=ForestGreen!75, line width=1pt, draw, minimum width=1pt, minimum height=1pt] (B23) [above=-39pt of B2]{};

	\draw[line width=1pt, black]
		(A12.center) -- (A22.center)
		(A13.center) -- (A23.center)
		
		(B11.center) -- (A21.center)
		(B13.center) -- (A23.center)
		
		(A11.center) -- (B21.center)
		(A13.center) -- (B23.center)

		(B11.center) -- (B21.center)
		(B12.center) -- (B22.center);

	\draw[line width=1pt, gray, dashed]
		(A11.center) -- (A21.center)
		(B12.center) -- (A22.center)
		(A12.center) -- (B22.center)
		(B13.center) -- (B23.center);

		\node[ellipse, black, fill=black!0, draw, line width=1pt, minimum width=30pt, minimum height=40pt] (C1) [right=1cm of V2]{};
			\node [circle, black, fill=Blue!50, line width=1pt, draw, minimum width=1pt, minimum height=1pt] (C11) [above=-15pt of C1]{};
			\node [circle, black, fill=blue!50, line width=1pt, draw, minimum width=1pt, minimum height=1pt] (C12) [above=-27pt of C1]{};
			\node [circle, black, fill=Blue!50, line width=1pt, draw, minimum width=1pt, minimum height=1pt] (C13) [above=-39pt of C1]{};
	
	\draw[line width=1pt, black]
		(C11.center) -- (A21.center)
		(C12.center) -- (A22.center);

		\node[ellipse, black, fill=black!0, draw, line width=1pt, minimum width=30pt, minimum height=40pt] (D1) [right=1cm of C1]{};
			\node [circle, black, fill=Blue!50, line width=1pt, draw, minimum width=1pt, minimum height=1pt] (D11) [above=-15pt of D1]{};
			\node [circle, black, fill=Blue!50, line width=1pt, draw, minimum width=1pt, minimum height=1pt] (D12) [above=-27pt of D1]{};
			\node [circle, black, fill=blue!50, line width=1pt, draw, minimum width=1pt, minimum height=1pt] (D13) [above=-39pt of D1]{};

	\draw[line width=0.5pt, gray, dashed]
		(C11.center) -- (D11.center)
		(C11.center) -- (D12.center)
		(C11.center) -- (D13.center)
		
		(C12.center) -- (D11.center)
		(C12.center) -- (D12.center)
		(C12.center) -- (D13.center)
		
		(C13.center) -- (D11.center)
		(C13.center) -- (D12.center)
		(C13.center) -- (D13.center);
		
	\draw[line width=1pt, black]
		(C11.center) -- (D11.center)
		(C11.center) -- (D13.center)
		
		(C12.center) -- (D12.center)
		(C12.center) -- (D13.center)
		
		(C13.center) -- (D11.center)
		(C13.center) -- (D13.center);
		
	\draw[line width=1.5pt, red!50]
		(C12.center) -- (D13.center);
		
		\node[ellipse, black, fill=black!10, draw, line width=1pt, minimum width=40pt, minimum height=100pt] (V5) [right=1cm of D1]{};
		\node[ellipse, black, fill=black!0, draw, line width=1pt, minimum width=30pt, minimum height=40pt] (E1) [above=-50pt of V5]{};
			\node [circle, black, fill=ForestGreen!75, line width=1pt, draw, minimum width=1pt, minimum height=1pt] (E11) [above=-15pt of E1]{};
			\node [circle, black, fill=ForestGreen!75, line width=1pt, draw, minimum width=1pt, minimum height=1pt] (E12) [above=-27pt of E1]{};
			\node [circle, black, fill=ForestGreen!75, line width=1pt, draw, minimum width=1pt, minimum height=1pt] (E13) [above=-39pt of E1]{};
			
		\node[ellipse, black, fill=black!0, draw, line width=1pt, minimum width=30pt, minimum height=40pt] (F1) [below=-50pt of V5]{};
			\node [circle, black, fill=Yellow, line width=1pt, draw, minimum width=1pt, minimum height=1pt] (F11) [above=-15pt of F1]{};
			\node [circle, black, fill=Yellow, line width=1pt, draw, minimum width=1pt, minimum height=1pt] (F12) [above=-27pt of F1]{};
			\node [circle, black, fill=Yellow, line width=1pt, draw, minimum width=1pt, minimum height=1pt] (F13) [above=-39pt of F1]{};

	\node[ellipse, black, fill=black!10, draw, line width=1pt, minimum width=40pt, minimum height=100pt] (V6) [right=1cm of V5]{};
		\node[ellipse, black, fill=black!0, draw, line width=1pt, minimum width=30pt, minimum height=40pt] (E2) [above=-50pt of V6]{};
			\node [circle, black, fill=ForestGreen!75, line width=1pt, draw, minimum width=1pt, minimum height=1pt] (E21) [above=-15pt of E2]{};
			\node [circle, black, fill=ForestGreen!75, line width=1pt, draw, minimum width=1pt, minimum height=1pt] (E22) [above=-27pt of E2]{};
			\node [circle, black, fill=ForestGreen!75, line width=1pt, draw, minimum width=1pt, minimum height=1pt] (E23) [above=-39pt of E2]{};
			
	\node[ellipse, black, fill=black!0, draw, line width=1pt, minimum width=30pt, minimum height=40pt] (F2) [below=-50pt of V6]{};
			\node [circle, black, fill=Yellow, line width=1pt, draw, minimum width=1pt, minimum height=1pt] (F21) [above=-15pt of F2]{};
			\node [circle, black, fill=Yellow, line width=1pt, draw, minimum width=1pt, minimum height=1pt] (F22) [above=-27pt of F2]{};
			\node [circle, black, fill=Yellow, line width=1pt, draw, minimum width=1pt, minimum height=1pt] (F23) [above=-39pt of F2]{};

	\draw[line width=0.5pt, gray, dashed]
		(E11.center) -- (E21.center)
		(E12.center) -- (E22.center)
		(E13.center) -- (E23.center)
		
		(F11.center) -- (E21.center)
		(F12.center) -- (E22.center)
		(F13.center) -- (E23.center)
		
		(E11.center) -- (F21.center)
		(E12.center) -- (F22.center)
		(E13.center) -- (F23.center)
		
		(F11.center) -- (F21.center)
		(F12.center) -- (F22.center)
		(F13.center) -- (F23.center);

	\draw[line width=1pt, gray, black]
		(E12.center) -- (E22.center)
		(E13.center) -- (E23.center)
		
		(F11.center) -- (E21.center)
		(F13.center) -- (E23.center)
		
		(E11.center) -- (F21.center)
		(E12.center) -- (F22.center)
		
		(F12.center) -- (F22.center)
		(F13.center) -- (F23.center);

	\draw[line width=1pt, black]
		(D12.center) to (F12.center)
		(D13.center) to (F13.center);
			
	\node [circle, black, fill=Blue!50, line width=1pt, draw, minimum width=1pt, minimum height=1pt] (s) [left=1cm of V1]{};
	\node [circle, black, fill=Blue!50, line width=1pt, draw, minimum width=1pt, minimum height=1pt] (t) [right=1cm of V6]{};

	\draw[line width=1pt, black]
		(s) -- (A11.center)
		(s) -- (A12.center);

\draw[line width=1pt, black]
		(t) -- (F23.center)
		(t) -- (F21.center);
		
	\draw[line width=1.5pt, blue!50]
		(s) -- (A12.center)
		(A12.center) -- (A22.center)
		(A22.center) -- (C12.center)
		
		(t) -- (F23.center)
		(D13.center) -- (F13.center)
		(F13.center) -- (F23.center);

	\node (p1) [above right=0.5cm and 0.35cm of V1]{$E_A$};
	\node (p2) [above right=0.5cm and 0.35cm of C1]{$E_1$};
	\node (p3) [above right=0.75cm and 0.15cm of s]{$E_B$};	
	\node (p33) [above right=-0.30cm and 0.35cm of V2]{$E_B$};	
	\node (p4) [above right=0.5cm and 0.35cm of V5]{$\EinvA$};
	\node (p5) [above left=0.75cm and 0.25cm of t]{$\EinvB$};	
	\node (p55) [above left=-0.30cm and 0.25cm of V5]{$\EinvB$};	
	
	\node (n1) [below=0.25cm of V1]{$V_1$};
	\node (n2) [below=0.25cm of V2]{$V_2$};
	\node (n3) [below=0.25cm of C1]{$V_3$};
	\node (n4) [below=0.25cm of D1]{$U_3$};
	\node (n5) [below=0.25cm of V5]{$U_2$};
	\node (n6) [below=0.25cm of V6]{$U_1$};
	\node (ns) [below=0.25cm of s]{$s$};
	\node (nt) [below=0.25cm of t]{$t$};

\end{tikzpicture}
   \caption{An illustration of the  input distribution $\distST$. Here, the directions of all edges are from left to right and hence omitted. The vertices $\sstar \in V_3$ and $\tstar \in U_3$ are marked blue and the potential edge $(\sstar,\tstar)$ is marked red--existence or non-existence of this edge uniquely determines whether or not $s$ can reach $t$ in $G$. 
   }
\label{fig:streach}
\end{figure*}
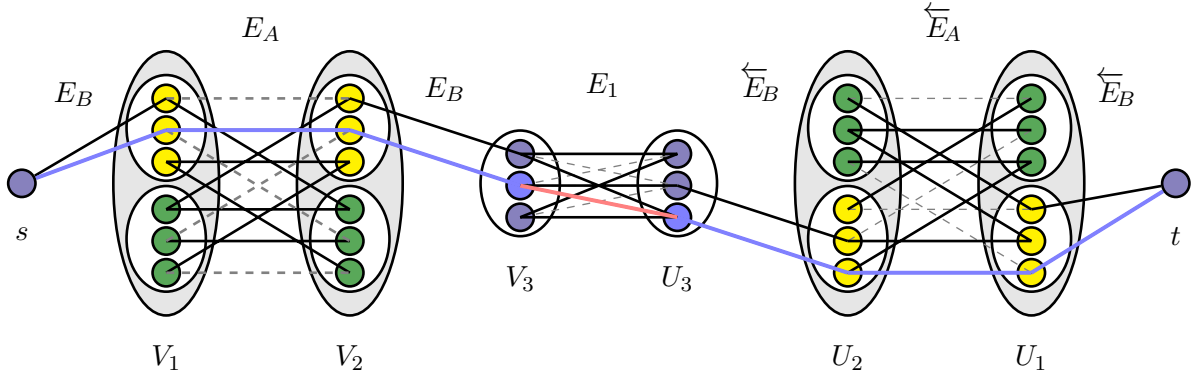

\begin{observation}\label{obs:distST}
	The following two remarks are in order:
	\begin{enumerate}[label=$(\roman*)$]
		\item The distributions of $E_1$, $H$, and $\Hinv$ are mutually independent in $\distST$. 
		\item $s$ can reach $t$ in $G$ iff the edge $(\sstar,\tstar) \in E_1$. \\
		(proof: the only vertex in $V_3$ reachable from $s$ is $\sstar$ and the only vertex in $U_3$ that reaches $t$ is $\tstar$, thus the only potential $s$-$t$ path is $s \leadsto \sstar \rightarrow \tstar \leadsto t$.)
	\end{enumerate}
\end{observation}

\subsection{Setup and Notation} 
Let $\protST$ be any \emph{deterministic} protocol for $\streach$ over the distribution $\distST$ 
with
\begin{align}
	\CC{\protST}{} = o(\eps^2 \cdot b^2), \label{eq:CCpST}
\end{align}
where $b := \frac{n}{2^{\Theta(\sqrt{\log{n}})}}$ is the parameter in~\Cref{thm:ureach} for instances of $\distUR$ and $\distURinv$. We will prove that the probability that $\protST$ outputs the correct
answer to $\streach$ is  $\frac{1}{2}+o(\eps)$, hence proving~\Cref{thm:streach} for deterministic protocols. The results for 
randomized protocols  follows immediately from this and an averaging argument (i.e., the easy direction of Yao's minimax principle~\cite{Yao83}). 

To facilitate our proofs, the following notation would be useful. For brevity, we use 
\begin{align*}
\Prot &:= (\Prot_{A1},\Prot_{B1},\Prot_{A2}), \tag{the messages communicated by Alice and Bob in both rounds of the protocol $\protST$} \\
Z_1 &:= (\Prot_{A1},\Prot_{B1},E_3), \tag{the \emph{extra} information known to Alice at the end of the first round of communication} \\
Z_2 &:= (\Prot,E_3,\sstar,\tstar). \tag{the \emph{extra} information known to Bob at the end of the second round of communication}
\end{align*}
We also use $O \in \set{0,1}$ to denote the output Bob at the end of the protocol. 

For any pair of vertices $v,u \in V_3 \times U_3$, 
we use the notation $E_1(v,u) \in \set{0,1}$ to denote whether or not the edge $(v,u) \in E_1$. For a fixed choice of $E_3 = E_B \cup \EinvB$ in $\distST$, we use $\Vstar_3$ and $\Ustar_3$ to denote the sets from which $\sstar$ and $\tstar$ are chosen uniformly at random
from conditioned on $E_B$ and $\EinvB$, respectively (see part $(i)$ of~\Cref{thm:ureach}). We also define: 
\[
E_1(\Vstar_3,\Ustar_3) := \set{E_1(v_i,u_i) \mid (v_i,u_i) \in \Vstar_3 \times \Ustar_3}.
\] 
We further assume a fixed arbitrary ordering of pairs $v,u \in V_3 \times U_3$ and define: 
\[
E_1^{<(v,u)} := E_1(v_1,u_1),E_1(v_2,u_2),\ldots
\]
for all pairs $(v_i,u_i) \in E_1(\Vstar_3,\Ustar_3)$ that appear before $(v,u)$ in this ordering (note that we ignore the other edges of $E_1$ that are not in $E_1(\Vstar_3,\Ustar_3)$ here).

Throughout the proof, we will do a couple of reductions from our lower bound for $\ureach$ (on the distribution $\distST$). To avoid ambiguity, in these cases, we use the notation ``Alice$_\ST$ and Bob$_\ST$'' to denote the players in the protocol $\protST$ 
for the $\streach$ problem and use ``Alice$_\UR$ and Bob$_\UR$'' to denote the players in the protocol for the $\ureach$ problem (similar notation will be used for $\ureachinv$ as well). As an aside, we emphasize that in these reductions 
 an ``Alice-player'' in one instance may need to play the role of a ``Bob-player'' in the other instance. 

\subsection*{Crucial Independence Properties} 

The following independence properties are crucial for our proofs. They are all based on 
the rectangle property of communication protocols and part $(i)$ of~\Cref{obs:distST}. 
\begin{align}
	&\rProt_{A2} \perp \rsstar,\rtstar \mid \rZ_1 \label{eq:ind1} \\
	&\rE_1 \perp \rsstar,\rtstar \mid \rZ_1,\rProt_{A2} \label{eq:ind2} \\
	&\rE_2 \perp \rE_1(\rsstar,\rtstar) \mid \rZ_1,\rZ_2 \label{eq:ind3}.
\end{align}
\begin{proof}[Proof of~\Cref{eq:ind1}]
	\begin{align*}
		\mi{\rsstar,\rtstar}{\rProt_{A2} \mid  \rZ_1} &= \mi{\rsstar,\rtstar}{\rProt_{A2} \mid  \rProt_{A1},\rProt_{A2},\rE_3} \tag{by the definition of $Z_1 = (\Prot_{A1},\Prot_{B1},E_3)$}\\
		&\leq \mi{\rE_2}{\rProt_{A2} \mid \rProt_{A1},\rProt_{B1},\rE_3} \tag{by the data processing inequality (\itfacts{data-processing}) as $\rsstar,\rtstar$ is fixed by $\rE_2$ conditioned on $\rE_3$}\\
		&\leq \mi{\rE_2}{\rE_1 \mid \rProt_{A1},\rProt_{B1},\rE_3} \tag{by the data processing inequality (\itfacts{data-processing}) as $\rProt_{A2}$ is fixed by $\rE_1$ conditioned on $\rE_{3},\rProt_{B1}$} \\
		&\leq \mi{\rE_2}{\rE_1 \mid \rProt_{A1},\rE_3} \tag{by~\Cref{prop:info-decrease} since $\rProt_{B1} \perp \rE_1 \mid \rE_2,\rProt_{A1},\rE_3$ as $\rE_2,\rProt_{A1}$ fixes $\rProt_{B1}$} \\
		&\leq \mi{\rE_2}{\rE_1 \mid \rE_3} \tag{by~\Cref{prop:info-decrease} since $\rProt_{A1} \perp \rE_2 \mid \rE_1,\rE_3$ as $\rE_1$ fixes $\rProt_{A1}$} \\
		&=0 \tag{by~\itfacts{info-zero} since $\rE_1 \perp \rE_2,\rE_3$ by part $(i)$ of~\Cref{obs:distST}}. 
	\end{align*}
	The proof now follows from~\itfacts{info-zero}. 
\end{proof}
\begin{proof}[Proof of~\Cref{eq:ind2}]
	\begin{align*}
		\mi{\rsstar,\rtstar}{\rE_1 \mid  \rZ_1,\rProt_{A2}} &\leq \mi{\rE_1}{\rE_2 \mid \rProt_{A1},\rProt_{B1},\rProt_{A2},\rE_3} \tag{by the data processing inequality (\itfacts{data-processing}) as $\rsstar,\rtstar$ is fixed by $\rE_2$ conditioned on $\rE_3$}\\
		&\leq \mi{\rE_1}{\rE_2 \mid \rProt_{A1},\rProt_{B1},\rE_3} \tag{by~\Cref{prop:info-decrease} since $\rProt_{A2} \perp \rE_2 \mid \rE_1,\rProt_{A1},\rProt_{B1},\rE_3$ as $\rE_1,\rProt_{B1},\rE_3$ fixes $\rProt_{A2}$} \\
		&=0, 
	\end{align*}
	as was shown in the previous proof. The proof now follows from~\itfacts{info-zero}. 
\end{proof}
\begin{proof}[Proof of~\Cref{eq:ind3}]
	\begin{align*}
		\mi{\rE_1(\rsstar,\rtstar)}{\rE_2 \mid \rZ_1,\rZ_2} &= \mi{\rE_1(\rsstar,\rtstar)}{\rE_2 \mid \rProt_{A1},\rProt_{B1},\rProt_{A2},\rE_3,\rsstar,\rtstar} \\
		&\leq \mi{\rE_1}{\rE_2 \mid \rProt_{A1},\rProt_{B1},\rProt_{A2},\rE_3,\rsstar,\rtstar} \tag{by the data processing inequality (\itfacts{data-processing}) as $\rE_1(\rsstar,\rtstar)$ is determined by $\rE_1 \mid \rsstar,\rtstar$}\\
		&\leq \mi{\rE_1}{\rE_2 \mid \rProt_{A1},\rProt_{B1},\rProt_{A2},\rE_3} \tag{by~\Cref{prop:info-decrease} since $\rsstar,\rtstar \perp \rE_1 \mid \rE_2,\rProt_{A1},\rProt_{B1},\rProt_{A2},\rE_3$ as $\rE_2,\rE_3$ fixes $\rsstar,\rtstar$} \\
		&=0,
		\end{align*}
	as was shown in the previous proof. The proof now follows from~\itfacts{info-zero}. 
\end{proof}

\subsection{Part One: The First Round of Communication}

In the following lemma, we prove that after the first round of the protocol, the (joint) distribution of $(\sstar,\tstar)$ conditioned on $Z_1 = (\Prot_{A1},\Prot_{B1},E_3)$ is almost the same as if we only conditioned on $E_3$. 
This is basically through a reduction from~\Cref{thm:ureach} considering $\sstar,\tstar$ are distributed (originally) according to $\distUR$ and $\distURinv$ and the public information $E_3$ provides the input of Bob in the instances of $\ureach$ and 
$\ureachinv$ in this reduction.  

\begin{lemma}\label{lem:first-round}
$\Ex_{\rZ_1} {\tvd{\distribution{\rsstar,\rtstar \mid Z_1}}{\distribution{\rsstar,\rtstar \mid E_3}}} = o(\eps).$
\end{lemma}
\noindent
We start with writing,
\begin{align}
	\text{LHS in~\Cref{lem:first-round}} &= \Ex_{\rZ_1} {\tvd{\distribution{\rsstar,\rtstar \mid Z_1}}{\distribution{\rsstar,\rtstar \mid E_3}}} \notag \\
	&\leq 	\Ex_{\rZ_1} {\tvd{\distribution{\rsstar \mid Z_1}}{\distribution{\rsstar \mid E_3}}} \notag \\
	&\hspace{3cm}+ 	\Ex_{\rZ_1}\,\Ex_{\rsstar \mid Z_1} {\tvd{\distribution{\rtstar \mid Z_1 , \sstar}}{\distribution{\rtstar \mid E_3,\sstar}}} \notag \\
	&= 	\Ex_{\rZ_1} {\tvd{\distribution{\rsstar \mid Z_1}}{\distribution{\rsstar \mid E_B}}} \notag \\
	&\hspace{3cm}+ 	\Ex_{\rZ_1}\,\Ex_{\rsstar \mid Z_1} {\tvd{\distribution{\rtstar \mid Z_1 , \rsstar}}{\distribution{\rtstar \mid \EinvB}}}, \label{eq:fr-final}
\end{align}
where the final equality holds because $E_3 = E_B \cup \EinvB$, $\rsstar \perp \rEinvB \mid E_B$ and $\rtstar \perp \rsstar,\rE_B \mid \EinvB$ by~\Cref{obs:distST}. In the following two claims, we bound each of the terms in \Cref{eq:fr-final} by $o(\eps)$.

\begin{claim}\label{clm:first-round-first}
	$\Ex_{\rZ_1} {\tvd{\distribution{\rsstar \mid \rZ_1}}{\distribution{\rsstar \mid E_B}}} = o(\eps).$
\end{claim}
\begin{proof}
	Consider the following \emph{one-way} protocol $\prot_{\UR}$ for instances $H := (\set{s} \sqcup V_1 \sqcup V_2 \sqcup V_3, E_{A} \sqcup E_{B})$ of $\ureach$ sampled from $\distUR$ between the two players Alice$_\UR$ and Bob$_\UR$: 
	\begin{enumerate}[label=$(\roman*)$]
		\item Alice$_\UR$ and Bob$_\UR$ sample $\EinvB \sim \distST \mid H$ using \underline{public randomness} (this is doable without either player knowing $H$ by part $(i)$ of~\Cref{obs:distST}). 
		\item Alice$_\UR$ samples $\EinvA \sim \distST \mid H,\EinvB$ and $E_1 \sim \distST \mid H,\EinvA,\EinvB$ 
		using \underline{private randomness} (again, this is doable by part $(i)$ of~\Cref{obs:distST}). 
		\item Alice$_\UR$ has access to the first-round input of Alice$_\ST$ and Bob$_\ST$ and thus can generate $\Prot_{A1}$ and $\Prot_{B1}$ and send them to Bob$_\UR$. 
		\item Bob$_\UR$ has access to $E_3$ and can compute $\distribution{\sstar \mid \Prot_{A1},\Prot_{B1},E_3}$. 
	\end{enumerate}
	
	Let $\delta \in (0,1)$ be the parameter for which $\prot_{\UR}$ internal $\delta$-solves $\ureach$ on $\distUR$. As such,  
	\begin{align*}
		\delta &= \Ex_{\rProt_{\UR},\rE_B}{\tvd{\distribution{\rsstar \mid \Prot_{\UR},E_B}}{\distribution{\rsstar \mid E_B}}} \tag{by~\Cref{def:eps-ureach}}\\
		&= \Ex_{\rProt_{A1},\rProt_{B1},\rEinvB,\rE_B}{\tvd{\distribution{\rsstar \mid \Prot_{A1},\Prot_{B1},\EinvB,E_B}}{\distribution{\rsstar \mid E_B}}} 
		\tag{as the public randomness in $\prot_{\UR}$ is $\EinvB$ and the message is $\Prot_{A1},\Prot_{B1}$} \\
		&= \Ex_{\rZ_1}{\tvd{\distribution{\rsstar \mid Z_1}}{\distribution{\rsstar \mid E_B}}} \tag{as $Z_1 = (\Prot_{A1},\Prot_{B1},E_{B},\Einv_{B})$} \\
		&= \text{LHS in~\Cref{clm:first-round-first}}.
	\end{align*}
	On the other hand, since 
	\begin{align*}
	\CC{\protUR}{} \leq \CC{\protST}{} = o(\eps^2 \cdot b^2) = o(\eps^2 \cdot n \cdot b) \tag{by~\Cref{eq:CCpST} and since $b < n$}, 
	\end{align*}
	by~\Cref{thm:ureach}, we should have that $\delta = o(\eps)$, hence proving the claim. 
\end{proof}

\begin{claim}\label{clm:first-round-second}
	$\Ex_{\rZ_1}\Ex_{\rsstar\mid\rZ_1} {\tvd{\distribution{\rtstar \mid \rZ_1,\sstar}}{\distribution{\rtstar \mid \EinvB}}} = o(\eps).$
\end{claim}

\begin{proof}
	The proof of this claim is  similar to that of~\Cref{clm:first-round-first} with some minor (yet crucial) changes.

	Consider the following \emph{one-way} protocol $\prot_{\UR}$ for instances $\Hinv := (U_3 \sqcup U_2 \sqcup U_1 \sqcup \set{t}, \Einv_A \sqcup \Einv_{B})$ of $\overleftarrow{\ureach}$ sampled from $\distURinv$ between the two players Alice$_\UR$ and Bob$_\UR$: 
	\begin{enumerate}[label=$(\roman*)$]
		\item Alice$_\UR$ and Bob$_\UR$ sample $E_B \sim \distST \mid \Hinv$ using \underline{public randomness} (this is doable without either player knowing $\Hinv$ by part $(i)$ of~\Cref{obs:distST}). 
		\item Alice$_\UR$ samples $E_A \sim \distST \mid \Hinv,E_B$ and $E_1 \sim \distST \mid \Hinv,E_A,E_B$ 
		using \underline{private randomness} (again, this is doable by part $(i)$ of~\Cref{obs:distST}). 
		\item Alice$_\UR$ has access to the first-round input of Alice$_\ST$ and Bob$_\ST$ and thus can generate $\Prot_{A1}$ and $\Prot_{B1}$ and send them to Bob$_\UR$. Additionally, Alice$_\UR$ computes the vertex $\sstar$ and sends that also to 
		Bob$_\UR$ (this is doable as Alice$_\ST$ has full information about $E_A \cup E_B$). 
		\item Bob$_\UR$ has access to $E_3$, and the vertex $\sstar$ and can thus compute $\distribution{\tstar \mid \Prot_{A1},\Prot_{B1},E_3,\sstar}$. 
	\end{enumerate}
	
	Let $\delta \in (0,1)$ be the parameter for which $\prot_{\UR}$ internal $\delta$-solves $\overleftarrow{\ureach}$ on $\distURinv$. As such, 
	\begin{align*}
		\delta &= \Ex_{\rProt_{\UR},\rEinvB}{\tvd{\distribution{\rtstar \mid \Prot_{\UR},\EinvB}}{\distribution{\rtstar \mid \EinvB}}} \tag{by~\Cref{def:eps-ureach}}\\
		&= \Ex_{\rProt_{A1},\rProt_{B1},\rsstar,\rE_B, \rEinvB}{\tvd{\distribution{\rtstar \mid \Prot_{A1},\Prot_{B1},E_B,\sstar,\EinvB}}{\distribution{\rtstar \mid \EinvB}}} 
		\tag{as the public randomness in $\prot_{\UR}$ is $E_B$ and the message is $\Prot_{A1},\Prot_{B1},\sstar$} \\
		&= \Ex_{\rZ_1,\rsstar}{\tvd{\distribution{\rtstar \mid Z_1,\sstar}}{\distribution{\rtstar \mid \EinvB}}} \tag{as $Z_1 = (\Prot_{A1},\Prot_{B1},E_{B},\Einv_{B})$}\\
		&= \text{LHS in \Cref{clm:first-round-second}}.
	\end{align*}
	On the other hand, since the vertex $\sstar$ communicated by Alice$_\UR$ requires additional $O(\log{n})$ bits on top of the message of $\protST$, we have, 
	\begin{align*}
	\CC{\protUR}{} \leq \CC{\protST}{} + O(\log{n}) = o(\eps^2 \cdot b^2) + O(\log{n}) = o(\eps^2 \cdot n \cdot b) \tag{by~\Cref{eq:CCpST} and since $b < n$ and $\eps > n^{-1/2}$ and so $\eps^2 \cdot n \cdot b \gg \log{n}$}. 
	\end{align*}
	By~\Cref{thm:ureach}, this implies that $\delta = o(\eps)$,  hence proving the claim. 
\end{proof}

\Cref{lem:first-round} now follows from plugging in the bounds in~\Cref{clm:first-round-first,clm:first-round-second} in~\Cref{eq:fr-final}.

\subsection{Part Two: The Second Round of Communication}

\Cref{lem:first-round} implies that the extra information $Z_1$ available to Alice at the beginning of the second round does not change the distribution of $(\sstar,\tstar)$ by much. 
We  use this to show that the message of Alice in the second round does not change the distribution of $E_1(\sstar,\tstar) \in \set{0,1}$ by much.

\begin{lemma}\label{lem:second-round}
$\Ex_{\rZ_1,\rZ_2}{\tvd{\distribution{\rE_1(\sstar,\tstar) \mid Z_1,Z_2}}{\distribution{\rE_1(\sstar,\tstar)}}} = o(\eps).$
\end{lemma}
As $Z_2 = (\Prot,E_3,\sstar,\tstar)$, $\Prot=(\Prot_{A1},\Prot_{B1},\Prot_{A2})$, and $E_3,\Prot_{A1},\Prot_{B1}$ are also part of $Z_1$, we have,  
	\begin{align}
		\text{LHS in~\Cref{lem:second-round}} &= \Ex_{\rZ_1,\rProt_{A2},\rsstar,\rtstar}{\tvd{\distribution{\rE_1(\sstar,\tstar) \mid Z_1,\Prot_{A2},\sstar,\tstar}}{\distribution{\rE_1(\sstar,\tstar)}}}  \notag \\
		&= \Ex_{\rZ_1,\rProt_{A2},\rsstar,\rtstar}{\tvd{\distribution{\rE_1(\sstar,\tstar) \mid Z_1,\Prot_{A2}}}{\distribution{\rE_1(\sstar,\tstar)}}} \tag{as $\rE_1(\sstar,\tstar) \perp \rsstar=\sstar,\rtstar=\tstar \mid Z_1,\Prot_{A2}$ by~\Cref{eq:ind2}} \\
		& = \Ex_{\rZ_1,\rProt_{A2}}\,\Ex_{\rsstar,\rtstar \mid Z_1,\Prot_{A2}}{\tvd{\distribution{\rE_1(\sstar,\tstar) \mid Z_1,\Prot_{A2}}}{\distribution{\rE_1(\sstar,\tstar)}}}. \label{eq:2nd-1}
	\end{align}

	Intuitively, by~\Cref{lem:first-round} and~\Cref{eq:ind1}, we know that $(\rsstar,\rtstar \mid Z_1,\Prot_{A2})$ is distributed almost uniformly over a set determined only by $E_3$ (on average over $Z_1,\Prot_{A2}$). Thus, in the following lemma, 
	we first check what happens to~\Cref{eq:2nd-1} if we switch this distribution to a truly uniform one, and then use this to conclude the proof.  
	\begin{lemma}\label{lem:rx-ry}
		Suppose $(\rx,\ry)$ is distributed as $(\rsstar,\rtstar) \mid \rE_3$, i.e., uniformly at random over $\Vstar_3 \times \Ustar_3$ for any $E_3$ (independent of the remaining variables). Then, 
		\begin{align*}
		\Ex_{\rZ_1,\rProt_{A2}}\,\Ex_{\rx,\ry \mid E_3}{\tvd{\distribution{\rE_1(x,y) \mid Z_1,\Prot_{A2}}}{\distribution{\rE_1(x,y)}}} = o(\eps). 
		\end{align*}
	\end{lemma}
	\begin{proof}
	We first have, 
	\begin{align*}
		\text{LHS in~\Cref{lem:rx-ry}} &= \Ex_{\rZ_1,\rProt_{A2}}\,\Ex_{\rx,\ry \mid E_3}{\tvd{\distribution{\rE_1(x,y) \mid Z_1,\Prot_{A2},x,y}}{\distribution{\rE_1(x,y)}}} \tag{as $\rE_1(x,y) \perp \rx=x,\ry=y \mid Z_1,\Prot_{A2}$ by the definition of $\rx,\ry$} \\
		&\leq \Ex_{\rZ_1,\rProt_{A2}}\,\Ex_{\rx,\ry \mid E_3}\bracket{\sqrt{\frac{1}{2} \cdot \kl{\distribution{\rE_1(x,y) \mid Z_1,\Prot_{A2},x,y}}{\distribution{\rE_1(x,y)}}}} 
		 \tag{by Pinsker's inequality~(\Cref{fact:pinskers})}\\
		&\leq  \sqrt{\frac{1}{2} \cdot \Ex_{\rZ_1,\rProt_{A2},\rx,\ry}\bracket{\kl{\distribution{\rE_1(x,y) \mid Z_1,\Prot_{A2},x,y}}{\distribution{\rE_1(x,y)}}}} 
		\tag{by concavity of $\sqrt{\cdot}$ and Jensen's inequality} \\
		&\qquad \qquad = \sqrt{\frac{1}{2} \cdot \mi{\rE_1(\rx,\ry)}{\rZ_1,\rProt_{A2},\rx,\ry}} \tag{by~\Cref{fact:kl-info}}.
	\end{align*}
	We now bound the mutual information term in the RHS above. Recall that $\Prot=(\Prot_{A1},\Prot_{B1},\Prot_{A2})$ and $Z_1 = (\Prot_{A1},\Prot_{B1},E_3)$. We have, 
	\begin{align*}
		\mi{\rE_1(\rx,\ry)}{\rZ_1,\rProt_{A2},\rx,\ry} &= \mi{\rE_1(\rx,\ry)}{\rProt,\rE_3,\rx,\ry} \\ 
		&=\mi{\rE_1(\rx,\ry)}{\rE_3,\rx,\ry} + \mi{\rE_1(\rx,\ry)}{\rProt \mid \rE_3,\rx,\ry} \tag{by the chain rule of mutual information (\itfacts{chain-rule})} \\
		&= \mi{\rE_1(\rx,\ry)}{\rProt \mid \rE_3,\rx,\ry}  \tag{as the first term above is zero by~\itfacts{info-zero} since $\rE_1(\rx,\ry) \perp \rE_3,\rx,\ry$ by~\Cref{obs:distST}} \\
		&= \Ex_{\rE_3} \bracket{\sum_{(x,y) \in \Vstar_3 \times \Ustar_3} \frac{1}{\card{\Vstar_3 \times \Ustar_3}} \cdot \mi{\rE_1(x,y)}{\rProt \mid E_3,x,y}} 
		\tag{as $\distribution{\rx,\ry \mid E_3}$ is uniform over $\Vstar_3 \times \Ustar_3$} \\
		&= \frac{1}{b^2} \cdot \Ex_{\rE_3} \bracket{\sum_{(x,y) \in \Vstar_3 \times \Ustar_3} \mi{\rE_1(x,y)}{\rProt \mid E_3,x,y}} \tag{by part $(i)$ of~\Cref{thm:ureach}, size of $\Vstar_3$ and $\Ustar_3$ is $b$} \\
		&= \frac{1}{b^2} \cdot \Ex_{\rE_3} \bracket{\sum_{(x,y) \in \Vstar_3 \times \Ustar_3} \mi{\rE_1(x,y)}{\rProt\mid E_3}}  \tag{as $\rx,\ry \perp \rE_1(x,y),\rProt \mid E_3$ } \\
		&\leq \frac{1}{b^2} \cdot \Ex_{\rE_3} \bracket{\sum_{(x,y) \in \Vstar_3 \times \Ustar_3}\mi{\rE_1(x,y)}{\rProt \mid \rE_1^{<(x,y)},E_3}}  
		\tag{by~\Cref{prop:info-increase} as $\rE_1(x,y) \perp \rE_1^{<(x,y)} \mid E_3$} \\
		&= \frac{1}{b^2} \cdot \Ex_{\rE_3} \bracket{\mi{\rE_1(\Vstar_3,\Ustar_3)}{\rProt \mid E_3}}  \tag{by the chain rule of mutual information (\itfacts{chain-rule})}\\
		&\leq \frac{1}{b^2} \cdot \log{\card{\supp{\rProt}}} \tag{by \Cref{fact:info-cc}} \\
		&= \frac{1}{b^2} \cdot \CC{\protST}{} \tag{as $\protST$ is a deterministic protocol and thus $\rProt$ only contains its message}. 
	\end{align*}
	Since $\CC{\protST}{} = o(\eps^2 \cdot b^2)$ by our assumption in~\Cref{eq:CCpST}, we obtain the proof of the lemma. 
\end{proof}

We can now conclude the proof of~\Cref{lem:second-round}. For any choice of $E_3$, define $\rx,\ry$ as random variables distributed as $(\rsstar,\rtstar) \mid E_3$ (as in the statement of~\Cref{lem:rx-ry}). We have, 
\begin{align*}
	\text{LHS in~\Cref{lem:second-round}} &= \Ex_{\rZ_1,\rProt_{A2}}\,\Ex_{\rsstar,\rtstar \mid Z_1,\Prot_{A2}}{\tvd{\distribution{\rE_1(\sstar,\tstar) \mid Z_1,\Prot_{A2}}}{\distribution{\rE_1(\sstar,\tstar)}}} \tag{by~\Cref{eq:2nd-1}} \\
	&\leq \Ex_{\rZ_1,\rProt_{A2}}\,\Ex_{\rx,\ry \mid Z_1,\Prot_{A2}}{\tvd{\distribution{\rE_1(x,y) \mid Z_1,\Prot_{A2}}}{\distribution{\rE_1(x,y)}}} \\
	&\hspace{1.6cm}+ \Ex_{\rZ_1,\rProt_{A2}}{\tvd{\distribution{\rsstar,\rtstar \mid Z_1,\Prot_{A2}}}{\distribution{\rx,\ry \mid Z_1,\Prot_{A2}}}} \tag{by~\Cref{fact:tvd-small}} \\
	&= \Ex_{\rZ_1,\rProt_{A2}}\,\Ex_{\rx,\ry \mid E_3}{\tvd{\distribution{\rE_1(x,y) \mid Z_1,\Prot_{A2}}}{\distribution{\rE_1(x,y)}}} \\
	&\hspace{1.6cm}+ \Ex_{\rZ_1,\rProt_{A2}}{\tvd{\distribution{\rsstar,\rtstar \mid Z_1,\Prot_{A2}}}{\distribution{\rx,\ry \mid E_3}}} \tag{as $\rx,\ry \perp \Prot \mid E_3$} \\
	&= o(\eps) + \Ex_{\rZ_1,\rProt_{A2}}{\tvd{\distribution{\rsstar,\rtstar \mid Z_1,\Prot_{A2}}}{\distribution{\rsstar,\rtstar \mid E_3}}} \tag{by~\Cref{lem:rx-ry} for the first term, and since $(\rx,\ry) \mid E_3$ is distributed the same as $(\rsstar,\rtstar)\mid E_3$} \\
	&= o(\eps) + \Ex_{\rZ_1}{\tvd{\distribution{\rsstar,\rtstar \mid Z_1}}{\distribution{\rsstar,\rtstar \mid E_3}}} \tag{as $\rsstar,\rtstar \perp \rProt_{A2}=\Prot_{A2} \mid Z_1$ by~\Cref{eq:ind1}} \\
	&= o(\eps) \tag{by~\Cref{lem:first-round}}.
\end{align*}
This concludes the proof of \Cref{lem:second-round}.

\begin{remark}\label{rem:internal-external}
	This is a good place to  state the reason we work with the notion of internal $\eps$-solving instead of external $\eps$-solving in~\cite{AssadiCK19b}. Roughly speaking, we would like the distribution of $(\sstar,\tstar)$ to 
	look ``almost random'' \emph{to Alice} at the beginning of the second round, so that her second-round message does not reveal much information about $E_1(\sstar,\tstar)$. But this has to hold	
	\emph{despite the fact} that she has the entire set $E_3$ at the beginning of the second round (and is thus ``internal'' to the underlying communication problem for $\ureach$ and its inner $\SI$ instance). 
	
	This issue was handled differently in~\cite{AssadiCK19b} by working with \emph{four players} instead of two in a way that the players were no longer internal to the problem they needed to solve. However, such a fix does \emph{not} work for our purpose: 
	\emph{even} if we did not provide $E_3$ directly to Alice, Bob could have sent it to Alice using $O(n)$ communication which is negligible, and thus making Alice again internal to the communication problem (in the setting of~\cite{AssadiCK19b}, such an 
	approach would require $\Omega(n^2)$ communication which is no longer negligible).   
\end{remark}

\subsection{Concluding the Proof of~\Cref{thm:streach}}\label{sec:conc-streach}

We are now ready to conclude the proof of~\Cref{thm:streach}. \Cref{lem:second-round} implies that conditioning on $Z_1,Z_2$ does not change the distribution of $E_1(\sstar,\tstar)$ by much. By the independence property 
of~\Cref{eq:ind3}, we know that this continues to hold even if we further condition on the input of Bob, i.e., $E_2$. We use this to prove that the probability that $\protST$ outputs the correct answer is almost the same as random guessing. 

\begin{claim}\label{clm:protST-errs}
	$\Pr\paren{\protST\textnormal{ outputs the correct answer}} = \frac{1}{2} + o(\eps)$. 
\end{claim}
\begin{proof}
Recall that $O \in \set{0,1}$ is the output of the protocol by Bob which is a function of $E_2$ and $Z_1,Z_2$, namely, the input to Bob and the information revealed to him (either
through $E_3$ or the protocol). Recall that  by part~$(ii)$ of~\Cref{obs:distST}, $E_1(\sstar,\tstar)$ determines the correct answer. 
Let $\BB(1/2)$ denote the Bernoulli distribution with mean $1/2$. We have, 
\begin{align*}
	\Pr\paren{\protST\textnormal{ outputs correctly}} &= \EX_{\rZ_1,\rZ_2,\rE_2}\,\,\Pr_{x\sim \rE_1(\sstar,\tstar) \mid Z_1,Z_2,E_2}\paren{O=x} \tag{note that conditioning on $\rZ_1,\rZ_2,\rE_2$ fixes $\rO=O$ but not $\rE_1(\sstar,\tstar)$}\\
	&\leq \EX_{\rZ_1,\rZ_2,\rE_2}\,\,\bracket{\Pr_{x \sim \BB(1/2)}\paren{O=x} + \tvd{\distribution{\rE_1(\sstar,\tstar) \mid Z_1,Z_2,E_2}}{\BB{(1/2)}}} \tag{by~\Cref{fact:tvd-small}}\\
	&\leq \frac{1}{2} + \EX_{\rZ_1,\rZ_2,\rE_2}\,\tvd{\distribution{\rE_1(\sstar,\tstar) \mid Z_1,Z_2,E_2}}{\BB{(1/2)}} \tag{as $O$ is fixed and $x$ is chosen uniformly at random from $\set{0,1}$} \\
	&= \frac{1}{2} + \EX_{\rZ_1,\rZ_2}\,\tvd{\distribution{\rE_1(\sstar,\tstar) \mid Z_1,Z_2}}{\BB{(1/2)}} \tag{as $\rE_1(\sstar,\tstar) \perp \rE_2=E_2 \mid Z_1,Z_2$ by~\Cref{eq:ind3}} \\
	&= \frac{1}{2} + \EX_{\rZ_1,\rZ_2}\,\tvd{\distribution{\rE_1(\sstar,\tstar) \mid Z_1,Z_2}}{\distribution{\rE_1(\sstar,\tstar)}} \tag{by definition of $\rE_1$} \\
	&= \frac{1}{2} + o(\eps). \tag{by~\Cref{lem:second-round}}
\end{align*}
This finalizes the proof. 
\end{proof}

\medskip

To conclude, we have shown that for any deterministic protocol $\protST$ with $\CC{\protST}{} = o(\eps^2 \cdot b^2)$, the probability that $\protST$ outputs the correct answer over the distribution $\distST$ is only $\frac{1}{2} + o(\eps)$. 
This can be extended directly to randomized protocols as by an averaging argument, we can always fix the randomness of any randomized protocol $\protST$ on the distribution $\distST$ to obtain a deterministic protocol with the same 
error guarantee. Noting that $b = \frac{n}{2^{\Theta(\sqrt{\log{n}})}}$ concludes the proof of~\Cref{thm:streach}.


\section{Graph Streaming Lower Bounds}\label{sec:stream}

We now obtain our graph streaming lower bounds by reductions from the $\streach$ communication problem defined in~\Cref{sec:streach}. 
The first step of all these reductions is to show that one can simulate any two-pass graph streaming algorithm on graphs $G=(V,E)$ using a protocol in the setting of the $\streach$ problem. 
The proof is via a standard simulation and is only provided for completeness considering the setting of $\streach$ is rather non-standard. 
	
\begin{proposition}\label{prop:stream}
	Any two-pass $S$-space streaming algorithm $\alg$ on graphs $G=(V,E_1\sqcup E_2\sqcup E_3)$ of $\streach$ can be simulated exactly by a communication protocol $\prot_{\alg}$ with $\CC{\prot_{\alg}}{} = O(S)$ and the communication-pattern restrictions 
	of the $\streach$ problem. 
\end{proposition}
\begin{proof}	
	Given $G=(V,E_1 \sqcup E_2 \sqcup E_3)$ of $\streach$, we define the stream $\sigma = E_1 \circ E_2 \circ E_3$.  
	We use $M_{i,j}$ for $i \in \set{1,2}$ and $j \in \set{1,2,3}$ as the \emph{memory content} of $\alg$ after the $i$-th pass over the set of edges $E_j$ in the stream $\sigma$. We note that since $\alg$ is a streaming algorithm, $M_{i,j}$ is 
	only a function of $M_{i,j-1}$ and $E_j$ with the exception that $M_{2,1}$ is a function of $M_{1,3}$ and $E_1$. The answer is also uniquely determined by $M_{2,3}$. 
	
	We now give the protocol $\prot_{\alg}$ that allows Alice and Bob  to simulate two passes of $\alg$ on $\sigma$:
	\begin{enumerate}[label=$(\roman*)$]
		\item Given $E_1$, Alice runs $\alg$ on $E_1$ and sends $M_{1,1}$ as the message $\Prot_{A1}$. 
		\item Given $\Prot_{A1} = M_{1,1}$ and $E_2$, Bob continues running $\alg$ on $E_2$ and sends $M_{1,2}$  to Alice as $\Prot_{B1}$. 
		This finishes the first round of communication in $\prot_{\alg}$ (but not the first pass of $\alg$ yet). 
		\item Given $\Prot_{B1}=M_{1,2}$ and $E_3$ at the beginning of the second round, Alice continues running $\alg$ on $E_3$ which concludes the first pass of $\alg$. Considering Alice also has $E_1$, she then continues the second pass
		by running $\alg$ on $E_1$ for the second time and sending $M_{2,1}$ as $\Prot_{A2}$ to Bob. 
		\item Given $\Prot_{A2}=M_{2,1}$ and $E_3$ now and having $E_2$ from before, Bob can compute both $M_{2,2}$ and $M_{2,3}$ and outputs the same answer as $\alg$. 
	\end{enumerate}
	It is straightforward to verify that communication cost of this protocol is $O(S)$ and it outputs the same exact answer as $\alg$. 
\end{proof}

\subsection{Directed Reachability}\label{sec:reach}

We obtain the following theorem for the directed reachability problem. 

\begin{theorem}[Formalization of~\Cref{res:reach}]\label{thm:reach}
	Any streaming algorithm that makes two passes over the edges of any $n$-vertex directed graph $G=(V,E)$ with two designated vertices $s,t \in V$ and outputs whether or not $s$ can reach $t$ 
	in $G$ with probability at least $2/3$ requires $\Omega(\frac{n^2}{2^{\Theta(\sqrt{\log{n}})}})$ space. 
\end{theorem}

\Cref{thm:reach} follows immediately from~\Cref{prop:stream} and our  lower bound in~\Cref{thm:streach}. 

We also present some standard extension of this lower bound to other problems related to the directed reachability problem. 

\begin{itemize}
	\item \textbf{Estimating number of vertices reachable from a source:} Consider any instance of the problem in~\Cref{thm:reach} and connect $t$ to $2n$ \emph{new vertices}. In the new graph, if $s$ can reach $t$,  then it can also reach at least $2n$ other vertices, 
	while if $s$ does not reach $t$, it can reach at most $n$ other vertices. Hence, the lower bound in~\Cref{thm:reach} extends to this problem as well which was studied (in a similar format) in~\cite{HenzingerRR98}. 
	
	\item \textbf{Testing if $G$ is acyclic or not:} Recall that the hard  distribution of graphs in~\Cref{thm:streach} and hence~\Cref{thm:reach} is supported on \emph{acyclic} graphs. If in these graphs, we connect $t$ to $s$ directly, then the graph remains acyclic iff $s$ cannot reach $t$. Hence, the lower bound in~\Cref{thm:reach} extends to this problem as well. 
	
	\item \textbf{Approximating minimum feedback arc set:} The lower bound for acyclicity implies the same bounds for any (multiplicative) approximation algorithm of  minimum feedback arc set  (the minimum number of edges to be deleted to make a 
	graph acyclic) studied in~\cite{ChakrabartiG0V20}. 
\end{itemize}

\subsection{Bipartite Perfect Matching}\label{sec:matching}

We obtain the following theorem for the bipartite perfect matching problem. 

\begin{theorem}[Formalization of~\Cref{res:matching}]\label{thm:matching}
	Any streaming algorithm that makes two passes over the edges of any $n$-vertex undirected bipartite graph $G=(L \sqcup R,E)$ and outputs whether or not $G$ has a perfect matching	with probability at least $2/3$ requires $\Omega(\frac{n^2}{2^{\Theta(\sqrt{\log{n}})}})$ space. 
\end{theorem}
\begin{proof}
	The proof is via a standard reduction from~\Cref{thm:reach}. Basically, we show that any algorithm for perfect matching problem can also be used to solve the $s$-$t$ reachability problem within the same asymptotic space. We note that this reduction is folklore
	and we claim no novelty in this part. 
	
	Consider a directed graph $H=(V_H,E_H)$ and two vertices $s,t \in V_H$. Create the following bipartite graph $G = (L \sqcup R,E)$: 
	\begin{itemize}
		\item For any vertex $v \in V_H \setminus \set{s,t}$, there are two vertices $v^{\ell} \in L$ and $v^{r} \in R$. We also have a new vertex $s^{\ell} \in L$ and $t^{r} \in R$ corresponding to $s$ and $t$. 
		\item For any (directed) edge $(u,v) \in E_H$, there is an (undirected) edge between $u^{\ell}$ and $v^{r}$ in $G$ (we assume without loss of generality that $s$ has no incoming edges and $t$ has no outgoing ones). 
		Moreover, for any $v \in V_H \setminus \set{s,t}$, there is an (undirected) edge between $v^{\ell}$ and $v^{r}$. 
	\end{itemize}  
	We now claim that $s$ can reach $t$ in $H$ iff $G$ has a perfect matching. To see this, first consider the matching $M = \set{(v^{\ell},v^{r}) \mid v \in V_H \setminus \set{s,t}}$ in $G$. In this matching, the only unmatched vertices are $s^{\ell}$ and $t^{r}$. 
	Now note that any \emph{augmenting path} of this matching $M$ in $G$ between $s^{\ell}$ and $t^{r}$ corresponds to a directed path from $s$ to $t$ in $H$. This implies that the only way for $G$ to have a perfect matching
	 is if there is a $s$-$t$ path in $H$ and vice versa. 
	
	As this reduction can be done ``on the fly'' in the streaming setting, the lower bound in~\cref{thm:matching} follows from~\Cref{thm:reach} immediately. 
\end{proof}

\subsection{Single-Source Shortest Path}\label{sec:sp}

Finally, we have the following theorem for the shortest path problem.

\begin{theorem}[Formalization of~\Cref{res:sp}]\label{thm:sp}
	Any streaming algorithm that makes two passes over the edges of any $n$-vertex undirected graph $G=(V,E)$ with two designated vertices $s,t \in V$ and outputs the length of the shortest $s$-$t$ path in $G$ 
	with probability at least $2/3$ requires $\Omega(\frac{n^2}{2^{\Theta(\sqrt{\log{n}})}})$ space. 
	
	The lower bound continues to hold even if the algorithm is allowed to output an estimate which, with probability at least $2/3$, is as large as the length of the shortest $s$-$t$ path and 
	strictly smaller than $\nicefrac{9}{7}$ times the length of the shortest $s$-$t$ path.  
\end{theorem}
\begin{proof}
	The proof is via a reduction from~\Cref{thm:streach} and applying~\Cref{prop:stream}. Consider any graph $H$ in the support of the distribution $\distST$ in~\Cref{thm:streach} and assume we simply make all 
	edges \emph{undirected} to obtain a graph $G$. We claim that the length of the shortest $s$-$t$ path in $G$ is $7$ if $s$ can reach $t$ in $H$, and is otherwise at least $9$. 
	
	The proof is as follows: $(i)$ when $s$ can reach $t$ in $H$, 
	we can simply take that path, which is of length $7$ to get from $s$ to $t$ in $G$; $(ii)$ in other case, we need to take at least one more back-edge (i.e., in $G$ go in the opposite direction of a directed edge in $H$) and one more forward edge to reach $t$ in 
	$G$, making the path to be of length at least $9$. See~\Cref{fig:streach} for an illustration. 
	
	The lower bound now follows from~\Cref{thm:streach} and~\Cref{prop:stream}. 
\end{proof}

\bibliographystyle{abbrv}
\bibliography{general}

\appendix
\part*{Appendix}

\section{Basic Information Theory Facts}\label{app:info}
Our proofs rely on basic concepts from information theory which we summarize below. We refer the interested reader to the excellent textbook by Cover and Thomas~\cite{CoverT06} for a broader introduction. 

\begin{fact}\label{fact:it-facts}
  Let $\rA$, $\rB$, $\rC$, and $\rD$ be four (possibly correlated) random variables.
   \begin{enumerate}
  \item \label{part:uniform} $0 \leq \en{\rA} \leq \log{\card{\supp{\rA}}}$, where $\supp{\rA}$ denotes the support of $\rA$. 
   \item \label{part:info-zero} $\mi{\rA}{\rB \mid \rC} \geq 0$. The equality holds iff $\rA \perp \rB \mid \rC$.
  \item \label{part:cond-reduce} 
    $\en{\rA \mid \rB,\rC} \leq \en{\rA \mid  \rB}$.  The equality holds iff $\rA \perp \rC \mid \rB$.
  \item \label{part:chain-rule}$\mi{\rA,\rB}{\rC \mid \rD} = \mi{\rA}{\rC \mid \rD} + \mi{\rB}{\rC \mid  \rA,\rD}$ (the chain rule of mutual information).
    \item\label{part:data-processing} $\mi{f(\rA)}{\rB \mid \rC} \leq \mi{\rA}{\rB \mid \rC}$ for any deterministic function $f$ (the data processing inequality). 
   \end{enumerate}
\end{fact}
\noindent
The above facts also immediately imply the following which we use  in our proofs. 
\begin{fact}\label{fact:info-cc}
$\mi{\rA}{\rB \mid \rC} \leq \en{\rA \mid \rC} \leq \en{\rA} \leq \log{\card{\supp{\rA}}}$. 
\end{fact}

\noindent
We will also use the following two standard inequalities regarding conditional mutual information. 

\begin{proposition}\label{prop:info-increase}
 If $\rA \perp \rD \mid \rC$, then, $\mi{\rA}{\rB \mid \rC} \leq \mi{\rA}{\rB \mid  \rC,  \rD}.$
\end{proposition}
 \begin{proof}
  By~\itfacts{cond-reduce}, since $\rA \perp \rD \mid \rC$, we have $\en{\rA \mid  \rC} = \en{\rA \mid \rC, \rD}$ and since conditioning can only decrease the entropy, $\en{\rA \mid  \rC, \rB} \geq \en{\rA \mid  \rC, \rB, \rD}$. As such,
	 \begin{align*}
	  \mi{\rA}{\rB \mid  \rC} &= \en{\rA \mid \rC} - \en{\rA \mid \rC, \rB} \leq \en{\rA \mid \rC, \rD} - \en{\rA \mid \rC, \rB, \rD} = \mi{\rA}{\rB \mid \rC, \rD}. \qed
	\end{align*}
	
\end{proof}

\begin{proposition}\label{prop:info-decrease}
 If $ \rA \perp \rD \mid \rB,\rC$, then, $\mi{\rA}{\rB \mid \rC} \geq \mi{\rA}{\rB \mid \rC, \rD}.$
\end{proposition}
 \begin{proof}
  By~\itfacts{cond-reduce}, since $\rA \perp \rD \mid \rB,\rC$, we have $\en{\rA \mid \rB,\rC} = \en{\rA \mid \rB,\rC,\rD}$ and since conditioning can only reduce the entropy, $\en{\rA \mid \rC} \geq \en{\rA \mid \rD,\rC}$. As such, 
  \begin{align*}
 	\mi{\rA}{\rB \mid  \rC} &= \en{\rA \mid \rC} - \en{\rA \mid \rB,\rC} \geq \en{\rA \mid \rD,\rC} - \en{\rA \mid \rB,\rC,\rD} = \mi{\rA}{\rB \mid \rC,\rD}. \qed
 \end{align*}
 
 \end{proof}

We shall also use the following measures of distance (or divergence) between distributions. 

\paragraph{KL-divergence.} For two distributions $\mu$ and $\nu$, the \emph{Kullback-Leibler divergence} between $\mu$ and $\nu$ is denoted by $\kl{\mu}{\nu}$ and defined as: 
\begin{align}
\kl{\mu}{\nu}:= \Ex_{a \sim \mu}\Bracket{\log\frac{\Pr_\mu(a)}{\Pr_{\nu}(a)}}. \label{eq:kl}
\end{align}
We have the following relation between mutual information and KL-divergence. 
\begin{fact}\label{fact:kl-info}
	For random variables $\rA,\rB,\rC$, 
	\[\mi{\rA}{\rB \mid \rC} = \Ex_{\rB,\rC}\Bracket{ \kl{\distribution{\rA \mid B,C}}{\distribution{\rA \mid C}}}.\] 
\end{fact}

\paragraph{Total variation distance.} We denote the total variation distance between two distributions $\mu$ and $\nu$ on the same 
support $\Omega$ by $\tvd{\mu}{\nu}$, defined as: 
\begin{align}
\tvd{\mu}{\nu}:= \max_{\Omega' \subseteq \Omega} \paren{\mu(\Omega')-\nu(\Omega')} = \frac{1}{2} \cdot \sum_{x \in \Omega} \card{\mu(x) - \nu(x)}.  \label{eq:tvd}
\end{align}
\noindent
\begin{fact}\label{fact:tvd-small}
	Suppose $\mu$ and $\nu$ are two distributions for a random variable $\rA$, then, 
	\[
	\Ex_{\mu}\bracket{\rA} \leq \Ex_{\nu}\bracket{\rA} + \tvd{\mu}{\nu} \cdot \max \card{\rA}.
	\]
\end{fact}

The following Pinsker's inequality bounds the total variation distance between two distributions based on their KL-divergence, 

\begin{fact}[Pinsker's inequality]\label{fact:pinskers}
	For any distributions $\mu$ and $\nu$, 
	$
	\tvd{\mu}{\nu} \leq \sqrt{\frac{1}{2} \cdot \kl{\mu}{\nu}}.
	$ 
\end{fact}

Finally, we use the following simple auxiliary result in our proofs. 

\begin{proposition}\label{prop:uniform-tvd}
	Let $n$ be an even integer, $\unif_n$ be the uniform distribution on $[n]$, and $\mu$ be any distribution on $[n]$ with $\tvd{\mu}{\unif_n} = \delta$. Suppose $S$ is the top half of elements 
	with the largest probability in $\mu$. Then, 
	$
	\Pr_{\mu}\paren{e \in S} \geq \frac{1}{2} + \frac{\delta}{2}.
	$ 
\end{proposition}
\begin{proof}
	Suppose without the loss of generality that $\mu(1) \geq \mu(2) \geq \cdots \geq \mu(n)$ by a renaming of elements. We thus have, 
	\begin{align}
	\Pr_{\mu}\paren{e \in S}  = \sum_{i=1}^{n/2} \mu(i). \label{eq:sum-i}
	\end{align}
	Let $\ell$ be the \emph{largest} index such that $\mu(\ell) \geq \frac{1}{n}$. By the definition in~\Cref{eq:tvd}, 
	\begin{align}
		\sum_{i=1}^{\ell} \paren{\mu(i) - \frac{1}{n}} = \delta. \label{eq:second-sum}
	\end{align}
	\paragraph{Case (a): when $\ell \leq n/2$.} We have, 
	\begin{align*}
		\sum_{i=1}^{n/2} \mu(i) &= \sum_{i=1}^{\ell} \mu(i) + \sum_{i=\ell+1}^{n/2} \mu(i) \geq \sum_{i=1}^{\ell} \mu(i) + \frac{n/2-\ell}{n-\ell} \cdot \sum_{i=\ell+1}^{n} \mu(i) \tag{as $\mu(\ell+1) \geq \mu(\ell+2) \geq \cdots \geq \mu(n)$} \\
		&= \sum_{i=1}^{\ell} \mu(i) + \frac{n/2-\ell}{n-\ell} \cdot \paren{1-\sum_{i=1}^{\ell} \mu(i)} =\sum_{i=1}^{\ell} \mu(i) \cdot \paren{\frac{n/2}{n-\ell}}+ \frac{n/2-\ell}{n-\ell} \\
		&\geq \paren{\frac{\ell}{n}+\delta} \cdot \paren{\frac{n/2}{n-\ell}}+ \frac{n/2-\ell}{n-\ell} \tag{by~\Cref{eq:second-sum}} = \frac{(n/2) \cdot \ell + (n^2/2) \cdot \delta + n^2/2 - n \cdot \ell}{n \cdot (n-\ell)} \\
		&= \frac{n-\ell + n \cdot \delta}{2 \cdot (n-\ell)} = \frac{1}{2} + \frac{\delta \cdot n}{2 \cdot (n-\ell)} \geq \frac{1}{2} + \frac{\delta}{2},
	\end{align*}
	as $\ell \leq n/2$. Plugging in this in~\Cref{eq:sum-i} proves the statement in this case. 
	\paragraph{Case (b): when $\ell > n/2$.} We have, 
	\begin{align*}
		\sum_{i=1}^{n/2} \mu(i) &\geq  \frac{n/2}{\ell} \cdot \sum_{i=1}^{\ell} \mu(i) \tag{as $\mu(1)  \geq \cdots \geq \mu(\ell)$, and by~\Cref{eq:second-sum}}
		\geq \frac{n/2}{\ell} \cdot \paren{\frac{\ell}{n}+\delta} 
		\geq \frac{1}{2} + \frac{\delta}{2},
	\end{align*}
	as $n/2 < \ell \leq n$. Plugging in this in~\Cref{eq:sum-i} proves  this case and concludes the entire proof. 
\end{proof}

\end{document}